\documentclass[11pt]{article}
\usepackage[left=1in, top=1in, right=1in, bottom=1in]{geometry}

\usepackage{amsmath, amssymb, amsthm, amsfonts}
\usepackage{mathabx}
\usepackage{natbib}
\usepackage[font=small,labelfont=bf]{caption}
\usepackage{hyperref}
\usepackage{algorithmic}
\usepackage{nicefrac}
\usepackage{graphicx}
\usepackage{xcolor}
\usepackage{bbm} 
\usepackage{listings}

\usepackage{multirow}
\usepackage{colortbl}
\usepackage{enumitem}
\usepackage{tablefootnote}
\usepackage{tikz}
\usepackage{pgfplots}
\pgfplotsset{compat=1.18}

\usepackage{titlesec}
\titleformat{\paragraph}[runin]{\normalfont\normalsize\bfseries}{}{0em}{}[.]

\allowdisplaybreaks 
\newtheorem{theorem}{Theorem}[section]
\newtheorem{mainresult}{Main Result}

\newtheorem{lemma}[theorem]{Lemma}
\newtheorem{proposition}[theorem]{Proposition}
\newtheorem{claim}[theorem]{Claim}
\newtheorem{observation}[theorem]{Observation}

\theoremstyle{definition}
\newtheorem{definition}[theorem]{Definition}
\newtheorem{example}[theorem]{Example}
\newtheorem{remark}[theorem]{Remark}

\usepackage[ruled]{algorithm2e}

\SetKwComment{Comment}{$\triangleright$ }{ }

\newcommand{\actions}{A}
\newcommand{\agents}{A} 
\newcommand{\outcomespace}{\Omega}
\newcommand{\outcome}{\omega}
\newcommand{\sstar}{S^\star}

\newcommand{\reals}{\mathbb{R}}
\newcommand{\indicator}[1]{\mathbbm{1}\left[{#1}\right]}

\newcommand{\cand}{\mathcal{C}}

\newcommand{\downscaling}{\Psi}
\newcommand{\optestimate}{\tilde{y}}
\newcommand{\optestimateround}[1]{\tilde{y}^{(#1)}}

\title{Multi-Agent Contracts\thanks{An extended abstract of this work appeared in the Proceedings of the 55th ACM Symposium on Theory of Computing (STOC 2023) \citep{DuettingEFK23}. The impossibility result for submodular rewards first appeared in an extended abstract in the Proceedings of the 36th ACM-SIAM Symposium on Discrete Algorithms (SODA 2025) \citep{DuettingEFK25}.}}

\author{Paul D\"utting\thanks{Google Research, Switzerland. Email: \texttt{duetting@google.com}} \and Tomer Ezra\thanks{Tel Aviv  University, Israel. Email: \texttt{tomerezra@tauex.tau.ac.il}} \and Michal Feldman\thanks{Tel Aviv University, Israel. Email: \texttt{mfeldman@tauex.tau.ac.il}}
\and Thomas Kesselheim\thanks{University of Bonn, Germany. Email: \texttt{thomas.kesselheim@uni-bonn.de}}}

\date{}

\begin{document}

\maketitle

\begin{abstract}
We study a natural combinatorial  single-principal multi-agent contract design problem, in which a principal motivates a team of agents to exert effort toward a given task.
At the heart of our model is a {\em reward function}, which maps the agent efforts to an expected reward of the principal.
We seek to design computationally efficient algorithms for finding optimal (or near-optimal) linear contracts for reward functions that belong to the complement-free hierarchy.

Our first main result gives constant-factor approximation algorithms for submodular and XOS reward functions, with value oracles for submodular reward functions and value and demand oracles for XOS reward functions.
It relies on an unconventional use of ``prices'' and (approximate) demand queries for selecting the set of agents that the principal should contract with, and exploits a novel scaling property of XOS functions and their marginals, which may be of independent interest.

As our second main result, we show that constant approximation is the best we can get for submodular reward functions, even with both value and demand oracles.
For the larger class of subadditive reward functions, we establish an $\Omega(\sqrt{n})$ impossibility for settings with $n$ agents.
A striking feature of this impossibility is that it applies to subadditive functions that are constant-factor close to submodular.
This rapid degradation presents a surprising departure from previous literature, e.g., on combinatorial auctions, where approximation guarantees tend to deteriorate more gracefully.
\end{abstract}

\section{Introduction}
\label{sec:intro}

Contract theory is a core problem in economic theory (c.f.~the 2006 ``Nobel Prize'' to Oliver Hart and Bengt Holmstr\"om), which seeks, using incentive mechanisms, to achieve desirable outcomes in the presence of unobservable actions.
It plays a major role in the design of markets for effort (or services), similar to the role that mechanism design and auction theory play in the design of markets for goods. 
Applications of contract design range from multi-million dollar markets, such as influencer marketing on social media platforms \citep[e.g.,][]{Statista}, to contracts for social goods, such as government-run programs to encourage afforestation/prevent deforestation \citep[e.g.,][]{LiIL21, AshlagiLL22}.

At its heart is the hidden-action principal-agent problem \citep[e.g.,][]{Holmstrom79,GrossmanHart83}, in which a principal seeks to incentivize an agent to take a costly action, whose stochastic outcome determines a reward for the principal. A contract defines monetary transfers from the principal to the agent based on the observable outcome. 
The principal's goal is to find the contract that maximizes her utility (expected reward minus transfer), when the agent chooses the action that maximizes his utility (expected transfer minus cost).

In its vanilla version, with a single principal and a single agent, the problem can be solved in polynomial time by solving one LP for each action \citep{GrossmanHart83}. While this approach can also be applied in more complex scenarios, its running time will usually be exponential in the (succinct) representation. Therefore, one has to understand the structure of optimal contracts in order to obtain computationally efficient algorithms for finding optimal contracts \citep[e.g.][]{BabaioffFNW12,DuttingRT21,DuttingEFK21}.

A natural extension of the classic single-principal single-agent model, are settings where a principal seeks to incentivize a \emph{team} of agents \citep{Holmstrom82}. In such scenarios, the complexity arises from the complex combinatorial structure of dependencies between the agent actions, and is already compelling with \emph{binary actions} (i.e., each agent can either exert effort or not) and \emph{binary outcome} (e.g., a project that can succeed or fail) \citep{BabaioffFNW12}. 
In this work, we provide computationally efficient algorithms and impossibilities for (approximating) the optimal contract in such multi-agent settings.

As in many economic models, the problem of optimally incentivizing teams of agents boils down to a clean and compelling combinatorial optimization problem. 
The challenge is that, even in cases where the dependencies between the agent actions exhibit nice structure, this structure does \emph{not} transfer to the objective~function.

\paragraph{Model -- Key ingredients}
To state and discuss our results, it will be useful to describe the key ingredients of the multi-agent hidden-action principal-agent problem \citep[c.f.][]{BabaioffFNW12} that we study in this paper.

Our focus is on the binary-action binary-outcome setting. In this model, a principal interacts with a set $\agents$ of $n$ agents. Agents have binary actions: They can either exert effort or not. Exerting effort comes with cost $c_i \in \reals_{\geq 0}$ for agent $i$. There are two possible outcomes $\outcomespace = \{0,1\}$, which we interpret as failure and success, respectively. The principal receives a reward of $r(1) \geq 0$ for success, and a reward of $r(0) = 0$ for failure. Without loss of generality, we assume that the reward for success is normalized to $1$.

A \emph{success probability function} $f: 2^\agents \rightarrow [0,1]$ maps each subset of agents $S \subset \agents$ that exert effort to a probability $f(S)$ of success. Note that since we normalize the principal's reward for success to $1$, $f(S)$ is also the expected reward of the principal. For this reason, we also refer to $f$ as the \emph{reward function}. 
We generally assume that $f$ is monotone so that $f(S) \leq f(T)$ for $S \subseteq T$, and normalized so that $f(\emptyset) = 0$. We assume that we can access $f$ through one or both of the following means: a \emph{value oracle}, which allows us to query $f(S)$ for any given $S$, or a \emph{demand oracle}, which on input $p_1, \ldots, p_n \in \reals_{\geq 0}$ returns $S$ that maximizes $f(S) - \sum_{j \in S} p_j$.

To incentivize the agents to exert effort, the principal designs a \emph{contract} $t: \Omega \rightarrow \mathbb{R}^n_{\geq 0}$, which defines payments $t_i(\outcome)$ for each agent $i \in \agents$ and possible outcome $\outcome \in \outcomespace$. In the binary-outcome case, it is without loss to focus on linear contracts  (Proposition~\ref{prop:binary-outcome}). A \emph{linear contract} is defined by a vector $\alpha = (\alpha_1, \ldots, \alpha_n) \in [0,1]^n$, with the interpretation that agent $i \in \agents$ should be paid $\alpha_i \cdot r (\outcome)$ for outcome $\outcome \in \outcomespace$. Each contract $\alpha$ induces a game among the agents, and we are interested in the (pure) Nash equilibria of that game.
We say that a linear contract incentivizes a set of agents $S \subseteq \agents$ to exert effort if for each agent $i \in S$ it holds that $\alpha_i \cdot f(S) - c_i \geq \alpha_i \cdot f(S\setminus\{i\})$, while for each agent  $i \not\in S$ it holds that $\alpha_i \cdot f(S) \geq \alpha_i \cdot f(S \cup \{i\}) -c_i$. 
The goal is to find a linear contract $\alpha$ and a set $S$ such that $S$ is incentivized by $\alpha$, and maximizes the principal's expected utility given by $g(S) = (1-\sum_{i \in \agents} \alpha_i) f(S)$.

We note that all our results extend to the problem of computing optimal linear contracts for more general outcome spaces (see Appendix~\ref{app:general-outcomes}). While linear contracts are no longer optimal in such settings, they can be shown to be max-min optimal when only the expected reward of each set of agents is known \citep{DuttingRT19,DuttingEFK21} (see Appendix~\ref{app:max-min}).

\subsection{Our Contribution}

We study the computational complexity of computing optimal and near-optimal linear contracts for the multi-agent hidden-action principal-agent problem for different classes of ``complement-free'' reward functions \citep{LehmannLN06}.

As it turns out, moving from a single-agent to a multi-agent setting significantly complicates matters. Even for the simplest class of reward functions --- additive reward --- it is NP-hard to compute the optimal contract; but the problem admits 
an FPTAS (see Appendix~\ref{app:additive}). As we progress in the hierarchy,  
the problem becomes  significantly more challenging.

Our first main 
result concerns submodular, and in fact, the more general class of XOS reward functions 
(for definitions see Section~\ref{sec:preliminaries}). 
We show that for both of these classes there exist polynomial-time 
algorithms that provide constant-factor approximation to the optimal contract, assuming either value oracles or value and demand oracles.

Our result relies on an unconventional use of ``prices'' and (approximate) demand queries for selecting the set of agents that the principal should seek to incentivize, and exploits a novel scaling property for XOS functions and their marginals, which may be of independent interest.

\begin{mainresult}[Theorem~\ref{thm:xos}] 
For any multi-agent setting with {\em submodular} reward function $f$, there exists a polynomial-time $O(1)$-approximation 
to the optimal contract,  that uses a {\em value} oracle.
If $f$ is {\em XOS}, an $O(1)$-approximation can be achieved using {\em demand} and {\em value} oracles.
\end{mainresult}

Our second main result shows that this result is tight. Specifically, we show that it is impossible to obtain a better than constant approximation for submodular functions with a sub-exponential number of queries, even with access to both demand and value oracles.

\begin{mainresult}[Theorem~\ref{thm:upper}] 
There exists a constant $\eta > 1$ 
for which any algorithm that uses a sub-exponential number of value and demand queries to the 
submodular reward function 
returns a $\eta$-approximate optimal contract with an exponentially small probability in $n$. 
\end{mainresult}

We also show that, the broader class of subadditive functions does not admit a constant-factor approximation. 
Specifically, for any polynomial-time algorithm, with demand or value oracles, there exists a multi-agent setting with  subadditive reward function $f$, such that the algorithm achieves no better than $\Omega\left(\sqrt{n}\right)$-approximation to the optimal contract, where $n$ is the number of agents (see Theorem~\ref{thm:subadditive-demand}). 
We also observe that there is a simple polynomial-time $O(n)$-approximation for subadditive rewards with value oracles (Remark~\ref{rem:subaddtive}). Table~\ref{tab:multi-agent} summarizes the state-of-the-art computational landscape of multi-agent contracts with binary actions, incorporating both our results and follow-up work (discussed below, in Section~\ref{sec:state-of-the-art}).

We remark that, while in this work we focus on the problem of computing pure Nash equilibria and we compare ourselves against the best pure Nash equilibrium, subsequent work has shown that there is at most a constant gap between the principal's utility under pure Nash equilibria and under more general equilibrium concepts such as mixed Nash equilibria and (coarse) correlated equilibria \citep{DuettingEFK25Lifting}. Therefore, all our approximation guarantees extend to guarantees against these stronger benchmarks. 

\begin{table}[t]
\centering
\fontsize{10pt}{12pt}
\selectfont
\scalebox{0.9}{
\begin{tabular}{|
>{\columncolor[HTML]{C0C0C0}}c |c
>{\columncolor[HTML]{EFEFEF}}c |
>{\columncolor[HTML]{EFEFEF}}c 
>{\columncolor[HTML]{FFFFC7}}c |}
\hline
\textbf{\begin{tabular}[c]{@{}c@{}}Multiple\\ agents\end{tabular}}    & \multicolumn{2}{c|}{\cellcolor[HTML]{C0C0C0}\textbf{Value Oracle}}                                                                                                                                                          & \multicolumn{2}{c|}{\cellcolor[HTML]{C0C0C0}\textbf{Value and Demand Oracle}}                                                                                                                                     \\ \hline
\multicolumn{1}{|l|}{\cellcolor[HTML]{C0C0C0}}                        & \multicolumn{1}{c|}{\cellcolor[HTML]{C0C0C0}\textbf{\begin{tabular}[c]{@{}c@{}}Upper bound \\ (pos)\end{tabular}}} & \cellcolor[HTML]{C0C0C0}\textbf{\begin{tabular}[c]{@{}c@{}}Lower bound\\ (neg)\end{tabular}}           & \multicolumn{1}{c|}{\cellcolor[HTML]{C0C0C0}\textbf{\begin{tabular}[c]{@{}c@{}}Upper bound \\ (pos)\end{tabular}}} & \cellcolor[HTML]{C0C0C0}\textbf{\begin{tabular}[c]{@{}c@{}}Lower bound\\ (neg)\end{tabular}} \\ \hline
\textbf{Additive}                                                     & \multicolumn{1}{c|}{\cellcolor[HTML]{FFFFC7}\begin{tabular}[c]{@{}c@{}}FPTAS\\ Proposition~\ref{prop:fptas-for-additive}\end{tabular}}                & \begin{tabular}[c]{@{}c@{}}OPT is\\ NP-hard\end{tabular}                                               & \multicolumn{1}{c|}{\cellcolor[HTML]{EFEFEF}FPTAS}                                                                 & \begin{tabular}[c]{@{}c@{}}OPT is\\ NP-hard\\ Proposition~\ref{prop:hardness-for-additive}\end{tabular}                            \\ \hline
\textbf{GS}
& \multicolumn{1}{c|}{\cellcolor[HTML]{EFEFEF}\begin{tabular}[c]{@{}c@{}}Constant \\ approx\end{tabular}}            & \begin{tabular}[c]{@{}c@{}}OPT is\\ NP-hard\end{tabular}                                               & \multicolumn{1}{c|}{\cellcolor[HTML]{EFEFEF}\begin{tabular}[c]{@{}c@{}}Constant\\ approx\end{tabular}}             & \cellcolor[HTML]{EFEFEF}\begin{tabular}[c]{@{}c@{}}OPT is\\ NP-hard\end{tabular}             \\ \hline
\textbf{\begin{tabular}[c]{@{}c@{}}Sub-\\ modular\end{tabular}}& \multicolumn{1}{c|}{\cellcolor[HTML]{FFFFC7}\begin{tabular}[c]{@{}c@{}}Constant\\ approx\\ Theorem~\ref{thm:xos}\end{tabular}}    & 
\multicolumn{1}{c|}{\cellcolor[HTML]{EFEFEF}\begin{tabular}[c]{@{}c@{}}No PTAS\\
\end{tabular}}  
& \multicolumn{1}{c|}{\cellcolor[HTML]{EFEFEF}\begin{tabular}[c]{@{}c@{}}Constant\\ approx\end{tabular}}             & 
\cellcolor[HTML]{FFFFC7}\begin{tabular}[c]{@{}c@{}}No PTAS\\ 
Theorem~\ref{thm:upper}\end{tabular}  
\\ \hline
\textbf{XOS}                                                          & \multicolumn{1}{c|}{\cellcolor[HTML]{EFEFEF} $O(n)$-approx}                                                                      & \cellcolor[HTML]{FFFFC7}\begin{tabular}[c]{@{}c@{}}No better\\ than $\Omega(n^{1/6})$\\ \citep{EzraFS24}\end{tabular} & \multicolumn{1}{c|}{\cellcolor[HTML]{FFFFC7}\begin{tabular}[c]{@{}c@{}}Constant\\ approx\\ Theorem~\ref{thm:xos}\end{tabular}}    & {\cellcolor[HTML]{EFEFEF}\begin{tabular}[c]{@{}c@{}}No PTAS\end{tabular}}                                    \\ \hline
\textbf{\begin{tabular}[c]{@{}c@{}}Sub-\\ additive\end{tabular}}& 
\multicolumn{1}{c|}{\cellcolor[HTML]{FFFFC7}\begin{tabular}[c]{@{}c@{}}$O(n)$-approx\\(Remark~\ref{rem:subaddtive}) \end{tabular}}
& 
\begin{tabular}[c]{@{}c@{}}No better\\ than $\Omega(n^{1/6})$\end{tabular}                                  & \multicolumn{1}{c|}{\cellcolor[HTML]{EFEFEF} $O(n)$-approx}                                                                      & \begin{tabular}[c]{@{}c@{}}No better\\ than $\Omega(n^{1/2})$\\ Theorem~\ref{thm:subadditive-demand}\end{tabular}               \\ \hline \hline
\textbf{\begin{tabular}[c]{@{}c@{}}Super-\\ modular\end{tabular}}     & 
\multicolumn{1}{c|}{\cellcolor[HTML]{FFFFFF}}                                                                     & \begin{tabular}[c]{@{}c@{}}No constant\\ approx\end{tabular}                                           & \multicolumn{1}{c|}{\cellcolor[HTML]{FFFFFF}}                                                                              & \begin{tabular}[c]{@{}c@{}}No constant\\ approx\\(if P $\neq$ NP)\\ {\small{\citep{VuongDPP23}}}\end{tabular}                      \\ \hline
\end{tabular}
}
\caption{
This table presents the state-of-the-art approximation results for multi-agent contract settings. 
The left part presents results under access to value oracle, and the right part presents results under access to both value and demand oracles. For each one we present both upper bounds (positive results) and lower bounds (negative results) on the achievable approximation.
The rows represent different reward function classes.
Yellow cells give the results, whereas gray cells represent results derived from other cells (where positive results carry over north (to sub-classes) and east (from value oracle to value and demand oracle), and negative results carry over south and west). 
}
\label{tab:multi-agent}
\end{table}

\subsection{Our Techniques}

In Section~\ref{sec:preliminaries}, we observe that the underlying optimization problem we need to solve is as follows. Given a cost $c_i \geq 0$ for each agent $i$ and a reward function 
$f: 2^\actions \to [0,1]$
our goal is to maximize the function 
$g: 2^\actions \to \{-\infty\} \cup (-\infty,1]$
defined by
\[
g(S) = \left( 1 - \sum_{i \in S} \frac{c_i}{f(i \mid S \setminus \{i\})} \right) f(S),
\]
where, for every $i \in S$, we use $f(i \mid S \setminus \{i\}) = f(S) - f(S \setminus \{i\})$ to denote the marginal value of agent $i$ with respect to set $S$.

The function $g$ gives the optimal principal's utility from incentivizing a given set of agents $S \subseteq A$, since $\alpha_i = c_i / f(i \mid S\setminus\{i\})$ for $i \in S$ and $\alpha_i = 0$ for $i \in A \setminus S$ are the smallest payments under which the agents in $S$ prefer to exert effort (see Proposition~\ref{prop:opt-alphas}).

The difficulty is that, even in cases where $f$ is highly structured, 
such as submodular, XOS, or subadditive, 
this structure does not carry over to $g$. 
For example, even in cases where $f$ is non-negative monotone 
and submodular, the induced $g$ 
will usually not be monotone and take negative values. If $f$ is only XOS, $g$ may even not be subadditive.

The following example illustrates this loss of structure for a simple setting with two agents and a submodular reward function.

\begin{example}[Multiple agents with submodular $f$]
Consider a setting with two agents $\agents = \{1,2\}$, with costs $c_1=c_2=0.25$, and  with the following submodular  
reward function $f$: $f(\emptyset) = 0$, $f(\{1\})=f(\{2\})=0.5$, and $f(\{1,2\})=0.75$.

To find the optimal contract, we can compute $g(S)$ for all $S \subseteq A$. If we don't want any of the agents to exert effort,
the best contract is $\alpha_1=\alpha_2=0$, for a principal's utility of $g(\emptyset)=0$. 
If we want to incentivize only agent 1 (resp., agent 2) to exert effort, the best 
contract is $\alpha_1=c_1/f(\{1\})=0.5$ and $\alpha_2=0$ (resp., $\alpha_1=0, \alpha_2=c_2/f(\{2\})=0.5$), for a principal's utility of $g(\{1\})=g(\{2\})=(1-\alpha_1)f(\{1\})=(1-\alpha_2)f(\{2\}) = 0.25$. Finally, 
if we want to incentivize
both agents to exert effort, the best contract is $\alpha_1=c_1/(f(\{1,2\})-f(\{2\}))=0.25/(0.75-0.5)=1$ and similarly $\alpha_2=1$, for a principal's utility of $g(\{1,2\})=(1-2)f(\{1,2\})<0$. 

Thus, the optimal contract is either $\alpha_1=0.5, \alpha_2=0$ or $\alpha_1=0, \alpha_2=0.5$, incentivizing
a single agent to exert effort.
Note that, while $f$ is monotone, non-negative,  and submodular, $g$ is non-monotone and takes negative values.
\end{example}

\paragraph{Constant-factor approximations for submodular and XOS}
To establish our algorithmic results, we show how to find a set $S$ such that $g(S) \geq \mathrm{constant} \cdot g(\sstar)$, where $\sstar$ is a set that maximizes $g(\sstar)$.
For the purpose of conveying the intuition behind our approach,  
assume in the following that $f(\sstar)$ is known to the algorithm 
(but not $\sstar$ itself)
and the reward contribution of a single agent is negligible, meaning that $f(\{i\})$ is tiny compared to $f(\sstar)$. In the technical sections, these assumptions will not be necessary.

A key ingredient in our proof is a pair of lemmas that draw connections between the value of the optimal solution and its marginals to the costs. Specifically, in Lemma~\ref{lem:sumci}, we show that $\sum_{i \in \sstar} \sqrt{c_i} \leq \sqrt{f(\sstar)}$. In Lemma~\ref{lem:stronger-condition}, we show that if for a set $S$ we have $f(i \mid S \setminus \{i\}) \geq \sqrt{2 c_i f(S)}$ for every $i \in S$, then $g(S) \geq \frac{1}{2} f(S)$.
These observations give rise to guiding the choice of our set $S$ by defining a ``price'' for each agent. 
Namely, let $p_i = \frac{1}{2} \sqrt{c_i f(\sstar)}$ for each agent $i$ and consider a \emph{demand set} $T$, i.e., a set that maximizes $f(T) - \sum_{i \in T}p_i$. 
We now have $f(T) \geq f(T) - \sum_{i \in T} p_i \geq f(\sstar) - \sum_{i \in \sstar} p_i \geq \frac{1}{2} f(\sstar)$, using the definition of a demand set and Lemma~\ref{lem:sumci}. 
By definition, the marginal contribution of every agent in the demand set must exceed its price, namely $f(i \mid T \setminus \{i\}) \geq p_i = \frac{1}{2} \sqrt{c_i f(\sstar)}$.  This condition looks almost like the one that is necessary to invoke Lemma~\ref{lem:stronger-condition}. However, note that we only have a lower bound on $f(T)$, no upper bound. Therefore possibly $f(T)$ is much larger than $f(\sstar)$. 

To deal with this, we establish a novel scaling property of XOS functions (Lemma~\ref{lem:xosscaling}). 
Our lemma shows that from every set $T$ one can remove 
elements so that its value reaches essentially any level, while the marginals of the remaining elements are kept sufficiently high with respect to their original marginals.
Namely, for every set $T$ and every $\downscaling < f(T)$, one can compute a subset $U \subseteq T$ such that 
$\frac{1}{2} \downscaling \leq f(U) \le \downscaling$ and $f(i \mid U \setminus \{ i \}) \geq \frac{1}{2} f(i \mid T \setminus \{ i \})$.
While this property is not too surprising for submodular functions, for XOS functions this is far from obvious, given the apparent lack of control over marginals, and may be of independent interest. Setting $\downscaling = \frac{1}{32}  f(\sstar)$, 
we can invoke Lemma~\ref{lem:stronger-condition} because with this choice 
$f(i \mid U \setminus\{i\}) \geq \frac{1}{2} f(i \mid T \setminus \{i\}) \geq \frac{1}{2} p_i = \frac{1}{4} \sqrt{c_i f(\sstar)} \geq \sqrt{2 c_i f(U)}$
and conclude that 
$g(U) \geq \frac{1}{2} f(U) \geq \frac{1}{128} f(\sstar) \geq \frac{1}{128} g(\sstar)$.

\paragraph{Inapproximability result for submodular}
Our impossibility proof for submodular functions under value and demand queries takes a different approach from previous works, particularly the hardness result of \citet{DuettingEFK23} for XOS functions with value and demand oracles and the hardness result of \citet{EzraFS24} for XOS functions with value oracles only.

In \citet{DuettingEFK23}, the key idea is to hide a single ``good'' set of agents with high value that cannot be efficiently learned using value and demand queries. Their construction ensures that all subsets of this set appear non-attractive, relying heavily on the non-monotonicity of XOS marginals. Since submodular functions do not exhibit this non-monotonicity, their approach does not extend to our setting.

The hardness result of \citet{EzraFS24}, in contrast, hides a small set of ``good'' agents with high value and high marginals. A crucial property of their construction is that this hidden set has significantly higher value than a random set of the same size. While effective under value queries, this property allows the hidden set to be learned relatively easily via demand queries, making their approach unsuitable for our problem.

Thus, proving hardness for submodular functions with access to both value and demand oracles requires new techniques. To this end, our proof introduces several novel ideas.

Our hard instance consists of a small group of ``good'' agents (of size $k = n/5$) and a larger group of ``bad'' agents (of size $n-k = 4n/5$), where the optimal contract incentivizes all good agents. The instance is constructed such that achieving better than a 
$1.03$-approximation requires identifying a set of at most 
$k$ elements containing at least $k/2$ good agents.

We first show that this task is impossible using value queries. A query on a large set provides no useful information, while a query on a small set only reveals whether the intersection with the good agents is ``large.'' However, very few small sets contain many good agents. The optimal set size for learning (i.e., one that is not too large and not too small) is shown to be $k$, 
leading to $\binom{n}{k}$ potential sets. Among these, at most $\binom{k}{k/2} \cdot \binom{n}{k/2}$ contain sufficiently many good agents, which is exponentially smaller than $\binom{n}{k}$.

The final step is to show that, by the particular structure of our instance, demand queries offer no additional power over value queries. Stated formally, any demand query in this instance can be simulated using a polynomial number of value queries.
Thus, a hardness result for value queries extends to demand queries as well.

\paragraph{Inapproximability result for subadditive}
The impossibility result for subadditive reward functions is more easily described for a non-normalized reward function. The proof can be adjusted to use a normalized reward function by dividing both the reward function and all the costs by $f(\actions)$.

To show the impossibility result for subadditive reward functions,
we first construct a subadditive function $f$ and costs $c_i \geq 0$ with the following properties: 
The function $f$ is such that $g(S) = O(1)$ for $\lvert S \rvert < \sqrt{n}$ and $g(S) \leq 0$ if $\lvert S \rvert \geq \sqrt{n}$. That is, the optimal principal utility is constant and any attempt to incentivize more than $\sqrt{n}$ agents would result in a negative (or zero) utility. 

We then slightly modify $f$ by choosing a random set $T^\star$ of size $\frac{n}{2}+1$, and increasing its value. Monotonicity and subadditivity of $f$ are preserved by this change. This change significantly increases the principal's utility from the set $T^\star$ to $g(T^\star) = \Omega(\sqrt{n})$. At the same time, it does not increase the utility from any other set, so that, not only is $T^\star$ the unique optimal set, but it is also the only one that approximates it well. The principal's problem thus boils down to finding $T^\star$. 

It remains to show that $T^\star$ cannot be found by a polynomial number of value or demand queries. 
Our argument for this exploits the fact that the modified function is almost identical to the original (symmetric) one to show 
that every value or demand
query reveals information only on a small set of candidates for $T^\star$.  
This then implies that
every algorithm that uses only a polynomial number of value and demand queries cannot guess $T^\star$ with high enough probability, resulting in an $\Omega(\sqrt{n})$ approximation to the principal's optimal utility.

A remarkable feature of this impossibility result is that it applies to a subadditive function that is constant-factor close to a submodular one (Observation~\ref{obs:subadditive-close}). This stands in contrast to prior work, e.g., on combinatorial auctions, where approximation guarantees tend to degrade more gracefully.

\subsection{Related Work}\label{sec:related-work}

\paragraph{Optimizing the effort of others}
Our work is part of an emerging frontier in Algorithmic Game Theory  on optimizing the effort of others  (see, e.g., the STOC 2022 TheoryFest workshop with the same title). This includes work on algorithmic contract design \citep[e.g.,][]{BabaioffFNW12,DuttingRT19,DuttingRT21}, strategic classification \citep[e.g.,][]{KleinbergR19,BechavodPWZ22}, optimal scoring rule design \citep[e.g.,][]{ChenY21,LiHSW22}, and delegation \citep[e.g.,][]{KleinbergK18,BechtelDP22}.

\paragraph{Combinatorial optimization and auctions} 
A number of fundamental papers explore 
combinatorial optimization problems with ``complement-free'' set functions. In a landmark paper, \citet{Feige09} gives constant-factor approximation algorithms for the welfare maximization problem in combinatorial auctions with submodular, XOS, and subadditive  bidders. An exciting line of work seeks to understand whether it is possible to match these bounds with truthful mechanisms, with the current state of the art being polyloglog approximations \citep{Dobzinski21,AssadiKS21}. Complement-free valuations also play a crucial role in combinatorial auctions with item bidding, where subadditive valuations enable constant-factor Price of Anarchy bounds  \citep{ChristodoulouKS16,BhawalkarR11,SyrgkanisT13,FeldmanFGL13}; and in the prophet inequalities/posted-price literature, where constant-factor approximations are known for XOS and subadditive valuations \citep{FeldmanGL15,DuttingFKL20,CorreaC23}. While the constant-factor approximation for XOS is via static item pricing \citep{FeldmanGL15}, the state of the art for such algorithms for subadditive is a loglog approximation \citep{DuttingFKL20}. There are also polynomial-time constant-factor approximation results for truthful revenue maximization with unit-demand bidders \citep{ChawlaHMS10}, additive bidders \citep{Yao15}, and XOS bidders \citep{CaiOZ22}.

\paragraph{Computational approaches to contracts} 
The current rise in interest in algorithmic approaches to contracts is motivated by the fact that more and more of the classic applications of contract theory are moving online, growing in scale, and happening in data-rich environments. 
We refer the reader to the recent survey of \citet{DuettingFT24} for a comprehensive overview of this literature.

Most relevant to us is the literature on \emph{combinatorial contracts}, which can be divided into three categories, depending on whether it concerns exponentially large outcome spaces \citep{DuttingRT21}, taking sets of  actions \citep{DuttingEFK21}, or settings with multiple agents \citep{BabaioffFNW12,BabaioffFN09,BabaioffFN10,EmekF12}.
The vast majority of research on combinatorial contracts has focused on the latter two dimensions (multiple actions and multiple agents) and is surveyed in detail in \citep{Feldman25}.

\citet{DuttingEFK21} study a 
single-principal single-agent setting, where 
the agent can take any subset of $n$ actions. The main result is a polynomial-time (in $n$) algorithm for gross-substitutes reward functions. They also show NP-hardness for general submodular reward functions (with value oracle access). 
While both \citep{DuttingEFK21} 
and the present paper deal with a reward function that maps any subset of $n$ actions to some expected 
reward, the induced optimization problems are fundamentally different: In the model of \citet{DuttingEFK21} a single agent chooses a subset of $n$ actions. 
A contract is defined by a single parameter $\alpha$, and the agent chooses a set of actions that is better than any other set of actions. In particular, not every set of actions can be incentivized. 
In the model we consider in this work, 
in contrast, each one of the $n$ agents makes a binary choice over actions (exert effort or not), thus an action profile corresponds to a subset of the $n$ agents. 
A contract is now defined by a vector $(\alpha_1,\ldots,\alpha_n)$. Typically, every set of agents can be incentivized, so there are exponentially many feasible solutions. The challenge is to find a feasible solution of high value.

\citet{BabaioffFNW12} and \citet{EmekF12} study the same model studied here---$n$ agents with binary actions.  
They assume that each agent succeeds in his individual task with a certain probability, depending on whether he exerts effort or not; 
then a Boolean function maps individual successes and failures to a success or failure of the project.
\citet{BabaioffFNW12} show that for Boolean functions represented by 
read-once networks the 
optimal contract problem is \#P-complete, and give a polynomial-time algorithm for AND networks. \citet{EmekF12} show that for OR networks (a special case of submodular) the problem is NP-hard, and give a FPTAS for this problem.
Our work significantly expands the landscape of this model, by considering general submodular, XOS, and subadditive reward~functions.

\subsection{Follow-Up Work}\label{sec:state-of-the-art}

Since the publication of the conference version of this paper, several studies have further mapped out the computational landscape of multi-agent contract settings. 
Specifically, \citet{EzraFS24} show that, for XOS reward functions, no algorithm that makes polynomially-many value queries can approximate the optimal contract (with high probability) to within a factor better than
$\Omega(n^{1/6})$.
In addition, \citet{VuongDPP23} show that, for supermodular reward functions, the optimal contract problem admits no polynomial-time constant-factor approximation algorithm nor an additive FPTAS, even with access to both value and demand queries (unless $\text{P} = \text{NP}$). 
In another recent follow-up work, \citet{QongEtAl24} consider a variation of our model in which the principal is constrained to contract with at most $k$ agents, and present a constant-factor approximation algorithm for XOS rewards with access to value and demand oracles.

Several follow-up studies have explored multi-agent contract settings, in which agents have richer action spaces. \citet{DuettingEFK25} consider a generalization of the model studied in this work, in which each agent can take any subset of a set of actions. Their main result is a constant-factor approximation for submodular rewards, when the algorithm has access to both value and demand oracles. This result is (asymptotically) best possible by the impossibility result that we give in Theorem~\ref{thm:upper}. \citet{DuettingEFK25Lifting} explore multi-agent settings with combinatorial action spaces under richer equilibrium notions, such as mixed Nash equilibria and (coarse) correlated equilibria. A key finding of their work is that for XOS rewards, there is at most a constant gap between the principal's utility under the best pure Nash equilibrium and the best (coarse) correlated equilibrium. \citet{DasarathaGS24} study a multi-agent contract setting in which agents have a continuum of actions, and characterize optimal payment schemes for such settings. \citet{CacciamaniBC024} also consider settings in which agents can take more than one action. They introduce the concept of a randomized contract, and show that an optimal randomized contract can be computed in time polynomial in the (explicit representation) size of the problem.

\citet{CastiglioniM023} and \citet{GoelCaruthersWade24} study an incomparable multi-agent model in which the agents' actions lead to observable individual outcomes, and payments can depend on the individual outcomes rather than on an aggregated global outcome as in our model. \citet{CastiglioniM023} explore approximation algorithms for IR-supermodular rewards and DR-submodular rewards. \citet{GoelCaruthersWade24} study a budgeted setting, where the principal seeks to maximize welfare subject to a budget constraint, and show that optimal contracts for this problem are so-called Luce contracts.
A different budgeted model is studied in \citep{FeldmanTPS25,Aharoni0T25,FeldmanGPS26budgets}, where the goal is to maximize an objective function—such as welfare, the principal’s utility, or reward—under a budget constraint on total payments to agents.

\citet{alon2025multi} have extended our model to settings with many projects, where the principal needs to
partition the agents among the projects, and within each project, the principal incentivizes the
agents through a contract.

\subsection{Organization}

We formally set up the model 
in Section~\ref{sec:preliminaries}. Our main positive result---the constant-factor approximation guarantees for submodular and XOS reward functions---%
appear in Section~\ref{sec:sm-and-xos}. 
We prove the impossibility results for submodular and subadditive reward functions in Section~\ref{sec:beyond-submodular}. We discuss how to extend our model and results to general outcome spaces in Appendix~\ref{app:general-outcomes} and Appendix~\ref{app:max-min}. Our results for additive reward functions 
appear in Appendix~\ref{app:additive}.

\section{Preliminaries}
\label{sec:preliminaries}

\paragraph{The multi-agent hidden-action principal-agent setting}
In our model, a single principal interacts with a set of $n$ agents $\agents = [n] = \{1, \ldots, n\}$. We focus on the following combinatorial binary-action and binary-outcome 
setup: Agents can either take action (exert effort) or not. Taking action comes with a cost $c_i \in \reals_{\geq 0}$ for agent $i$. There are two possible outcomes $\outcomespace = \{0,1\}$, which we interpret as failure ($\outcome = 0$) and success ($\omega = 1$), respectively. Each outcome $\outcome \in \outcomespace$ is associated with a reward $r(\omega) \geq 0$ for the principal. We assume that the principal's reward for success is normalized so that $r(1) = 1$, and that the reward for failure is $r(0) = 0$.

A \emph{success probability function} $f: 2^\agents \rightarrow [0,1]$ maps each set of agents $S \subseteq \agents$ that exert effort to a probability $f(S)$ of success. Note that since we normalize the principal's reward for success to $1$, $f(S)$ is also the expected reward of the principal. For this reason, we also refer to $f$ as the \emph{reward function}.
We generally assume that the reward function $f$ is monotone non-decreasing so that for any two sets of agents $S, S'$ with $S \subseteq S' \subseteq \agents$ it holds that $f(S) \leq f(S')$. We also assume that the reward function $f$ is normalized in the sense that $f(\emptyset) = 0$. We write $f(i \mid S) := f(S \cup \{i\}) - f(S)$ for the marginal contribution of $i \in \agents$ to $S \subseteq \agents$.

An important feature of the model is that the agents' actions are \emph{hidden}. That is, the principal cannot directly observe the actions chosen by the agents, only the outcome, which is determined stochastically based on the actions.

\paragraph{Moral hazard and contracts}

A main challenge in our problem is what economists refer to as \emph{moral hazard}: In and by itself the agents have no interest in exerting effort, as exerting effort is costly and the benefits from that effort go to the principal.

To incentivize the agents to exert effort, the principal designs a \emph{contract} $t: \outcomespace \rightarrow \reals_{\geq 0}^n$, specifying a non-negative payment $t_i(\omega)$ to agent $i$ for each possible outcome $\omega \in \Omega$. Note that in this general form, a contract may specify positive payments $t_i(1), t_i(0)$ for both success and failure. The requirement that payments should be non-negative is a standard assumption in the contracts literature, known as \emph{limited liability}.
 
A particularly popular class of contracts in practice are so-called linear (or commission-based) contracts.
A \emph{linear contract} is defined by a vector $\alpha = (\alpha_1, \ldots, \alpha_n) \in \reals^n_{\geq 0}$, and sets $t_i(\omega) = \alpha_i r(\omega)$; namely, $t_i(1)= \alpha_i$ and $t_i(0) = 0$.
Thus, when $S \subseteq \agents$ is the set of agents that exert effort, the expected transfer to agent $i$ is $\alpha_i f(S)$. Note that this only depends on the expected reward $f(S)$, and not the details of the distribution.

\paragraph{Utility functions and equilibria}

Consider contract $t$, and let $S$ be the set of agents that exert effort. Let $T_i(S,t)$ denote the expected payment to agent $i\in A$, i.e., $T_i(S,t) = f(S) t_i(1) + (1-f(S)) t_i(0)$, and let $T(S,t)$ denote the expected payment to all agents, i.e., $T(S,t) = \sum_{i \in \agents} T_i(S,t)$. Then the \emph{principal's utility} is given by $u_P(S,t) = f(S) - T(S,t)$; while agent $i$'s utility is given by $u_i(S,t) = T_i(S,t) - \indicator{i \in S} \cdot c_i$, where $\indicator{i \in S} = 1$ if $i \in S$ and $\indicator{i \in S} = 0$. 

Note that this means that agent $i$ is paid $T_i(S,t)$
irrespective of whether $i \in S$, while the cost $c_i$ is only incurred when $i \in S$ (i.e., agent $i$ exerts effort). 

Each contract $t$ induces a game among the agents, and we are interested in the (pure) Nash equilibria of that game. Formally, we say that contract $t$ \emph{incentivizes} the set of agents $S \subseteq A$ to exert effort if 
\begin{align*}
u_i(S,t) &\geq u_i(S \setminus\{i\},t) && \text{for all $i \in S$, and}\\
u_i(S,t) &\geq u_i(S\cup \{i\},t) && \text{for all $i \not\in S$.}
\end{align*}

In general games, pure Nash equilibria need not exist. However, in follow-up work,    \citet{VuongDPP23,DuettingEFK25} showed that for each linear contract $\alpha$, the induced game is a potential game, and hence every $\alpha$ admits at least one pure Nash equilibrium.

In the binary-action binary-outcome setting that we consider here, we show a much stronger property, namely that each non-redundant set of agents (i.e., sets of agents that do not contain agents with zero marginal contribution but positive cost) can be incentivized. Moreover, the optimal contract to incentivize any such set is a linear contract.

\begin{proposition}
\label{prop:binary-outcome}
\label{prop:opt-alphas}
   (a) A set of agents $S \subseteq \agents$ can be incentivized by some contract if and only if no agent $i \in S$ simultaneously has $f(i \mid S \setminus \{i\}) = 0$ and $c_i > 0$. 
   If a set of agents $S \subseteq \agents$ can be incentivized by some contract then it can also be incentivized by a linear contract.
   (b) Among all contracts that incentivize set $S \subseteq \agents$ the one that maximizes the principal's utility is the following linear contract:
    \begin{align*}
    &\alpha_i = \frac{c_i}{f(S) - f(S \setminus \{i\})} = \frac{c_i}{f(i \mid S \setminus \{i\})} &&\text{for all $i \in S$ s.t. $f(i \mid S \setminus\{i\}) > 0$, and}\\
    &\alpha_i = 0 &&\text{otherwise.}
    \end{align*}
\end{proposition}

\begin{proof}
Consider a general (not necessarily linear) contract $t$. Let's write $\alpha_i = t_i(1)$ and $\beta_i = t_i(0)$. A set of agents $S$  is incentivized by contract $t$ if and only if
\begin{align}
    &\alpha_i f(S) + \beta_i (1-f(S)) - c_i \geq \alpha_i f(S \setminus \{i\}) + \beta_i (1-f(S \setminus \{i\}))
    &&\text{for all $i \in S$, and}  \label{eq:ic-one}\\
    &\alpha_i f(S) + \beta_i (1-f(S))\geq \alpha_i f(S \cup \{i\}) + \beta_i (1-f(S \cup \{i\})) - c_i 
    &&\text{for all $i \not\in S$}. \label{eq:ic-two}
\end{align}

For claim (a) observe that if there exists an agent $i \in S$ such that $f(i \mid S \setminus\{i\}) = 0$ and $c_i > 0$ then no $t$ (resp.~$\alpha_i,\beta_i$) can satisfy 
Inequality~\eqref{eq:ic-one}
because the inequality is equivalent to $- c_i \ge 0$. 
On the other hand, if there is no such agent, then for each agent $i \in S$ either $f(i \mid S \setminus\{i\}) > 0$ or $c_i = 0$. In this case, for $i \in S$ with $f(i \mid S \setminus\{i\}) > 0$  
Inequality~\eqref{eq:ic-one} 
can be satisfied by choosing $\beta_i = 0$ and $\alpha_i$ large enough; while for any $i \in S$ with $c_i = 0$ or any $i \not \in S$ the respective Inequalities~\eqref{eq:ic-one} and \eqref{eq:ic-two} can be satisfied by choosing $\alpha_i = \beta_i = 0$.

For claim (b) observe that the principal's utility $u_P(S,t)$ is non-increasing in all $\alpha_i$ and $\beta_i$. So, for $i \in S$ such that $f(i \mid S \setminus \{i\}) > 0$, it is optimal to set $\beta_i = 0$ and then the lower bound on $\alpha_i$ in Inequality~\eqref{eq:ic-one} becomes  $\alpha_i \geq c_i/f(i \mid S \setminus \{i\})$. On the other hand, for $i \in S$ such that $c_i = 0$ as well as any $i \not\in S$ we minimize payments by setting $\alpha_i = \beta_i = 0$.
\end{proof}

\paragraph{The contract design problem}

We consider the problem of maximizing the principal's utility, when the principal can propose both the contract $t$ and the set of agents $S \subseteq A$ that exert effort.

By Proposition~\ref{prop:opt-alphas}, the principal can restrict attention to linear contracts, and 
the problem reduces to a purely combinatorial problem:
\[
\max_{S \in 2^\agents} g(S) \quad\text{where}\quad g(S) := \left(1-\sum_{i \in S} \frac{c_i}{f(i \mid S \setminus \{i\})} \right) f(S),
\]
and we let $c_i/f(i \mid S \setminus \{i\}) = 0$ when $c_i = 0$ and $f(i \mid S \setminus \{i\}) = 0$, and $c_i/f(i \mid S \setminus \{i\}) = \infty$ when $c_i > 0$ and $f(i \mid S \setminus \{i\}) = 0$.

Let $\sstar$ be the optimal choice of agents, i.e., the set that maximizes $g$. We say that $S$ is a \emph{$\gamma$-approximation} (where $\gamma \geq 1$) if $\gamma \cdot g(S) \geq g(\sstar)$.

\paragraph{Classes of reward functions $f$} 

We focus on  reward functions $f: 2^\agents \to [0,1]$ that belong to one of the following classes of complement-free set functions \citep{LehmannLN06}: 

\begin{itemize}
    \item Set function $f$ is \emph{additive} if there exist values $v_1, \ldots, v_n \in \reals_{\geq 0}$ such that $f(S)=\sum_{i\in S} v_i$. 
    \item Set function $f$
    is \emph{submodular} if for any two sets $S, S' \subseteq \agents$ with $S \subseteq S'$ and any $i \in \agents$ it holds that $f(i \mid S) \geq f(i \mid S')$.
    \item Set function $f$
    is \emph{XOS} if there exists a collection of additive functions 
    $\{a_i: 2^\agents \rightarrow \reals_{\geq 0}\}_{i = 1, \ldots, k}$ 
    such that for each set $S \subseteq \agents$ it holds that $f(S) = \max_{i = 1, \ldots, k} a_i(S) = \max_{i = 1, \ldots, k} \sum_{j \in S} a_{ij}$. 
    Given an XOS function $f$ and a set $S \subseteq \agents$, there exists an additive function $a_i$ such that $a_i(S)=f(S)$ and $a_i(T) \leq f(T)$ for all $T \subseteq \agents$; this function is called the additive supporting function of $f$ on $S$.
    
    \item Set function $f$
    is \emph{subadditive} if for any two sets $S, S' \subseteq \agents$ it holds that $f(S) + f(S') \geq f(S \cup S')$.
\end{itemize}

It is well known that $\text{submodular} \subset \text{XOS} \subset \text{subadditive}$ and all containment relations are strict \citep{LehmannLN06}.

\paragraph{Primitives for accessing $f$}
As is common in the combinatorial optimization literature involving set functions, we assume two primitives for accessing $f$: 
\begin{itemize}
    \item A \emph{value oracle} for $f$ is given $S \in 2^\agents$ and returns $f(S)$.
    \item A \emph{demand oracle} for $f$ is given a vector of prices $p = (p_1, \ldots, p_n) \in \reals^n_{\geq 0}$ 
    and returns a set $S \in 2^\agents$ that maximizes $f(S) - \sum_{j \in S} p_j$.
\end{itemize}

Both value and demand queries are considered standard in combinatorial optimization problems over set functions. 
In markets for goods (e.g., combinatorial auctions), a demand query 
returns an optimal
bundle to purchase given item prices. 
In our combinatorial contracting setting, a demand query 
returns an optimal set of actions 
the principal would choose if she executed the actions herself.
Demand oracles have proven useful in prior work on combinatorial contracts; see \citep{DuttingEFK21}.

\paragraph{Auxiliary lemma for XOS functions}

For our main result on XOS reward functions we need the following lemma, which generalizes a well-known property of submodular functions to XOS functions.

\begin{lemma}\label{lem:xos-sum-of-marginals-extended}[Cf.~\citet[][Lemma 1]{fu2012conditional}]\label{lem:xos-sum-of-marginals}
For any XOS function $f$ and any sets $S\subseteq T$,
\[
\sum_{i \in S} f(i \mid T \setminus\{i\}) \leq f(S).
\]
\end{lemma}
\begin{proof}
Let $a$ be an additive supporting 
function for $f$ on $T$ so that $a(T) = f(T)$ and $a(T') \leq f(T')$ for all $T' \subseteq T$. Then, for any $i \in T$, it holds that
\[
f(i \mid T \setminus\{i\}) = f(T) - f(T \setminus\{i\}) \leq  a(T) - a(T \setminus\{i\}) = a(\{i\}).
\]
Summing over all $i \in S$ we obtain 
\[
\sum_{i \in S} f(i \mid T \setminus\{i\}) \leq \sum_{i \in S} a(\{i\}) = a(S) \leq  f(S),
\]
as claimed.
\end{proof}

\begin{remark}[Extension to more general equilibrium notions]
We note that, while it is natural to focus on pure Nash equilibria, the principal may strictly benefit from inducing a mixed rather than a pure Nash equilibrium, even for submodular rewards \citep[][Example 3.1]{BabaioffFN10}. 
However, recent work by \citet{DuettingEFK25Lifting} shows that for submodular and XOS rewards, the principal loses at most a constant factor by restricting attention to pure Nash equilibria. Remarkably, this guarantee holds not only relative to mixed Nash equilibria, but even relative to the more permissive notion of coarse correlated equilibria.
\end{remark}

\section{Constant Factor for Submodular and XOS}
\label{sec:sm-and-xos}

In this section, we present our main positive results: polynomial-time constant-factor approximation algorithms 
for submodular and XOS multi-agent combinatorial contracts.

\begin{theorem}\label{thm:xos}
The following hold:
\begin{enumerate}
\item For submodular $f$ and $n$ agents, it is possible to compute an $O(1)$-approximation to the optimal contract in polynomial time using value queries.
\item For XOS $f$ and $n$ agents, 
it is possible to compute an $O(1)$-approximation to the optimal contract  in polynomial time using value and demand queries. 
\end{enumerate}
\end{theorem}

Recall that we use $g:2^\actions \to \{-\infty\} \cup (-\infty,1]$ 
to denote the  principal's utility as a function of the set of incentivized agents. I.e., 
    $g(S) := f(S) \left(1-\sum_{i\in S} \frac{c_i}{f(i\mid S\setminus\{i\})}\right)$. 
Let $\sstar$ be the optimal set of agents, i.e., the set that maximizes $g$. 

Below we present the full argument for submodular/XOS reward functions assuming value and demand oracle access to the reward function. 
The result to submodular reward functions with only value oracle access requires only small modifications and relies on known  algorithms for computing approximate demand sets \citep{SviridenkoVW17,Harshaw19a}. We defer the details of this extension to Appendix~\ref{app:extension-to-submodular}.

\subsection{Decomposing the Benchmark}

Our first lemma provides a useful decomposition of the benchmark by showing that  $g(\sstar)$ is upper bounded by the sum of $f(S^{\star} \cap \agents')$, where $\agents' = \{ i \in \agents \mid \frac{c_i}{f(\{i\})} \leq \frac{1}{2} \}$, and $\max_{i \in A} g(\{i\})$ --- the best contract for incentivizing a single agent.

This may look innocent, but is not, because---as we already observed earlier---generally none of the nice structural properties of $f$ (such as non-negativity, monotonicity, submodularity or being XOS)  carry over to $g$.

An important consequence of the lemma is that, since it's easy to find the best contract for incentivizing a single agent, we can focus on the non-trivial task of finding a contract 
whose principal's utility is in $\Omega(f(\sstar \cap \actions'))$.

\begin{lemma}
\label{lem:onebigagent}
If $f$ is XOS (or, more generally, subadditive),  then 
\[
g(\sstar) \leq f(\sstar \cap \agents') + \max\{0,\max_{i \in \agents} g(\{i\})\}.
\]
\end{lemma}

\begin{proof}
If $g(\sstar) = 0$, then the claim is trivial.
Otherwise, we first prove that $|\sstar \setminus \agents'| \leq 1$.
This is since, 
\begin{eqnarray*}
0  & < & g(\sstar) = f(\sstar)  \left(1- \sum_{i\in \sstar} \frac{c_i}{f(i\mid\sstar\setminus\{i\})}\right) \\ & \leq &  f(\sstar) \left(1- \sum_{i\in \sstar\setminus\actions'} \frac{c_i}{f(\{i\})}\right) \leq f(\sstar) \left(1-  \frac{|\sstar\setminus\actions'|}{2}\right),
\end{eqnarray*}
where the second inequality follows from subadditivity of $f$ which implies that $f(i \mid \sstar \setminus \{i\}) = f(\sstar) - f(\sstar \setminus \{i\}) \leq f(\{i\})$.
Since $f(\sstar)>0$, this implies that $ |\sstar\setminus\actions'| \leq 1$.

If $|\sstar \setminus \agents'| = 0$, the statement follows since  $g(\sstar) = g(\sstar\cap \actions') \leq f(\sstar\cap \actions')$.
Else, let $i^\star$ be the single item in $\sstar \setminus \agents'$. We have 
\begin{eqnarray*}
g(\sstar) & \leq & f(\sstar \cap \agents') + f(\{i^\star\}) \left( 1 - \frac{c_{i^\star}}{f(i^\star \mid \sstar \setminus \{i^\star\})} \right) \\ & \leq & f(\sstar \cap \agents') + \max_{i \in \agents} g(\{i\}).
\end{eqnarray*}
The first inequality follows from subadditivity of $f$, implying that 
$f(\sstar) \leq f(\sstar \cap \agents') + f(\{i^\star\})$, 
and by decreasing the payments. The second inequality follows again by subadditivity of $f$, implying that $f(i^\star \mid \sstar \setminus \{i^\star\}) \leq f(\{i^\star\})$. 
This concludes the proof.
\end{proof}

\subsection{Relaxing the Problem}
We next present two crucial lemmas that draw connections between rewards, marginal rewards, and costs. This pair of lemmas relaxes the problem and motivates our approach for finding a good set of agents to incentivize in a contract via prices and demand queries.

Clearly, for the optimal set of agents we must have $\sum_{i \in S^{\star}} c_i \leq f(S^{\star})$ because otherwise the reward cannot compensate the incurred cost. In the multi-agent hidden-action setting with XOS rewards, we can strengthen this observation as follows.

\begin{lemma}
If $f$ is XOS, then for all $S \subseteq \sstar$ we have \[
\sum_{i\in S} \sqrt{c_i}  \leq \sqrt{f(S)}.
\] \label{lem:sumci} \label{lem:stronger-sumci}
\end{lemma}
\begin{proof}
Note that $f(i \mid \sstar \setminus \{i\}) > 0$ for all $i \in \sstar$ with $c_i > 0$ because otherwise $g(\sstar) = -\infty$, whereas $g(\emptyset) = 0$, contradicting the optimality of $\sstar$.

Let's first consider the case where $f(\sstar) = 0$. In this case, we have $f(i \mid \sstar \setminus \{i\}) \leq f(\{i\}) \leq f(\sstar) = 0$ for all $i \in \sstar$. This means that $c_i = 0$ for all $i \in \sstar$, implying the statement.

So consider the case where $f(\sstar) > 0$. By optimality $g(\sstar) \geq 0$, so
\begin{align*}
 & g(\sstar)  = f(\sstar) \left(1-\sum_{i\in \sstar} \frac{c_i}{f(i\mid \sstar \setminus\{i\})}\right) \geq 0 .
\end{align*}
Using $f(\sstar) >0$, we obtain
\[
\sum_{i\in \sstar} \frac{c_i}{f(i\mid \sstar \setminus\{i\})} \le 1.
\]

For $i \in S$ let $x_i = \sqrt{\frac{c_i}{f(i\mid \sstar \setminus\{i\})}}$ 
and $y_i = \sqrt{f(i\mid \sstar \setminus\{i\})}$. This is well-defined because $\frac{c_i}{f(i\mid \sstar \setminus\{i\})} \geq 0$ by our initial observation. The Cauchy-Schwarz inequality states 
\[
\left(\sum_{i \in S} x_i y_i \right)^2 \leq \left(\sum_{i \in S} x_i^2 \right) \left(\sum_{i \in S} y_i^2 \right).
\]
Using this we obtain
\[
\left(\sum_{i \in S} \sqrt{c_i} \right)^2 \leq  \underbrace{\left(\sum_{i \in S}\frac{c_i}{f(i\mid \sstar \setminus\{i\})}\right)}_{\leq 1} \left(\sum_{i \in S} f(i\mid \sstar \setminus\{i\}) \right) \leq f(S),
\]
where the last inequality holds by Lemma~\ref{lem:xos-sum-of-marginals}. Taking the square root on both sides of the inequality establishes the claim.
\end{proof}

The next lemma shows that if for some set $S$, the marginal of every agent $i$ is not too small, then the principal pays at most half of the reward as transfers. Therefore, our approach will be to find a set $S$ for which $f(S)$ is high and also all marginals fulfill these constraints.
\begin{lemma}\label{lem:stronger-condition}
If $f$ is XOS, then for any set $S$ that fulfills $f(S) > 0$ and
\[
f( i \mid S \setminus \{ i \}) \geq \sqrt{2 c_i f(S)} \quad \text{for all $i \in S$,}
\]
we have $g(S) \geq \frac{1}{2} f(S)$. 
\end{lemma}

\begin{proof}
Consider any $i \in S$. If $c_i > 0$, note that we have to have $f( i \mid S \setminus \{ i \}) > 0$. So
\[
f( i \mid S \setminus \{ i \}) \geq \sqrt{2 c_i f(S)}
\]
is equivalent to 
\[
\frac{c_i}{f( i \mid S \setminus \{ i \})} \leq \frac{1}{2} \frac{f( i \mid S \setminus \{ i \})}{f(S)}.
\]
For any $i \in S$ with $c_i = 0$, we defined $\frac{c_i}{f( i \mid S \setminus \{ i \})} = 0$, even when the denominator is zero; so also $\frac{c_i}{f( i \mid S \setminus \{ i \})} \leq \frac{1}{2} \frac{f( i \mid S \setminus \{ i \})}{f(S)}$ because $f( i \mid S \setminus \{ i \}) \geq 0$. Summing over all $i \in S$ we obtain
\[
\sum_{i \in S} \frac{c_i}{f( i \mid S \setminus \{ i \})} \leq \frac{1}{2} \frac{\sum_{i \in S} f( i \mid S \setminus \{ i \})}{f(S)} \leq \frac{1}{2},
\]
where the second inequality holds by  Lemma~\ref{lem:xos-sum-of-marginals}. 
Therefore
\[
g(S) = \left( 1 - \sum_{i \in S} \frac{c_i}{f( i \mid S \setminus \{ i \})}\right) f(S) \geq \frac{1}{2} f(S).\qedhere
\] \end{proof}

\subsection{A Scaling Property of XOS Functions and Their Marginals}

The other crucial ingredient in our argument for finding  a contract is a novel scaling property of XOS functions and their marginals, which roughly says that we can scale down the reward $f(T)$ of any set $T$ to whatever level we wish, while also ensuring that the marginals of the elements that remain stay high (with respect to the original marginals).

This property is not too suprising for submodular $f$, and is indeed easy to obtain for this class of functions: Consider iteratively dropping elements from $T$ one by one to get down to the desired level. Then, by submodularity, the marginals of the items that remain can only go up.
For XOS $f$, however, this is far from obvious (given the apparent lack of control over marginals), and may be of independent interest.

\begin{lemma}
\label{lem:xosscaling}
Let $f:2^\actions\rightarrow \reals_{\geq 0}$ 
be an XOS function. Given a set $T \subseteq\actions$, and parameters $\delta \in (0,1]$ and $0\leq \downscaling < f(T)$, Algorithm~\ref{alg:scaling} runs in polynomial time with value oracle access to $f$ and finds a set $U \subseteq T$ such that 
\begin{equation}
 (1-\delta)\downscaling \leq f(U) \leq \downscaling + \max_{i\in T} f(\{i\}), \label{eq:prop_val} 
\end{equation}
and 
\begin{equation}
    f(i \mid U \setminus \{i\}) \geq \delta f(i \mid T \setminus \{i\}) \qquad \text{for all $i \in U$}. \label{eq:prop_marg}
\end{equation}
\end{lemma}

\begin{algorithm}[t]
\caption{Scaling sets for XOS}\label{alg:scaling}
\KwData{An XOS function $f:2^\actions \rightarrow \reals_{\geq 0}$, a set $T$, parameters $0 \leq \downscaling < f(T)$ and $\delta\in (0,1]$
}

\KwResult{A set $U \subseteq T$ satisfying Equations \eqref{eq:prop_val} and \eqref{eq:prop_marg}}

    Let $T_0$ be an inclusion-wise minimal subset of $T$ with $f(T_0) = f(T)$  \quad
    \Comment{\mbox{I.e.,  $f(S)< f(T_0), \forall S \subsetneq T_0$}}

    \For{ $t=1,\ldots,|T_0|$}{ 
    Let $i_t = \arg\min_{i \in T_{t-1}} \frac{f(i\mid T_{t-1} \setminus \{i\})}{f(i\mid T_0 \setminus \{i\})}$ 
    
    Let $T_{t} =T_{t-1} \setminus \{i_t\}$
    
    Let $\delta_t= \frac{f(i_t\mid T_{t})}{f(i_t\mid T_0 \setminus \{i_t\})}$ 
    }

    Let $j^\star = \max \{j \mid f(T_{j}) > \downscaling \} $

    Let $k^\star = \min \{k \mid f(T_k) \leq (1-\delta) f(T_{j^\star}) \} $

    Let $t^\star = \arg\max_{t\in \{j^\star+1,\ldots,k^\star\}} \delta_t $ 
    
    \Return $U=T_{t^\star-1}$
    
\end{algorithm}

Before delving into the proof, let us first gain a clearer understanding of the lemma statement and Algorithm~\ref{alg:scaling}.
As already mentioned, for submodular functions $f$, the property claimed in Lemma~\ref{lem:xosscaling} is straightforward. The set $U$ can be obtained by dropping elements from $T$ in any order up to the point where the function value would drop below $\Psi$. 
The set obtained this way would fulfill $\downscaling \leq f(U) \leq \downscaling + \max_{i\in T} f(\{i\})$ and $f(i \mid U \setminus \{i\}) \geq f(i \mid T \setminus \{i\})$ for all $i \in U$. For general XOS functions, such a set $U$ might not exist. Therefore, these two inequalities are relaxed by introducing a parameter $\delta \in (0, 1]$. Note that as $\delta$ decreases \eqref{eq:prop_val} becomes stronger while \eqref{eq:prop_marg} becomes weaker.

To find a set fulfilling \eqref{eq:prop_val} and \eqref{eq:prop_marg}, Algorithm~\ref{alg:scaling} proceeds as follows: The algorithm starts from a set $T_0 \subseteq T$ that has the same value as $T$, but in which all marginals are positive. Then it obtains sets $T_1, T_2, \ldots$ by repeatedly removing one element. In each step, it removes the element whose marginal to the current set is smallest relative to the marginal the element had in the original set. Let's denote the element that is removed in step $t$ by $i_t$, so that $T_t = T_{t-1} \setminus \{i_t\}$ for all $t = 1, \ldots, |T_0|$. Then the element $i_t$ that is removed in step $t$ minimizes $f(i \mid T_{t-1}\setminus\{i\})/f(i\;|\;T_0\setminus\{i\})$ among all $i \in T_{t-1}$. The claim is that following this process there is a step in which $f(T_t)$ is between $(1-\delta)\downscaling$ and $\downscaling + \max_{i\in T} f(\{i\})$ and no marginal decreased by more than a $\delta$-factor.

The intuitive reason is as follows: The amount by which $f(T_t)$ decreases compared to $f(T_{t-1})$ is exactly the marginal $f(i_t \mid T_{t-1}\setminus\{i_t\})$. 
If this marginal is tiny, then the decrease is tiny as well. But given that we are considering a sequence of steps in which the overall value decreases by at least $\delta \Psi$, not every decrease can be tiny. In particular, there must be a step in which an item $i_t$ is removed whose current marginal $f(i_t \mid T_{t-1} \setminus \{i_t\})$ relative to the original marginal $f(i_t \mid T_0 \setminus\{i_t\})$ was sizable. Moreover, since we always remove the element with the smallest ratio between current marginal and original marginal, at that point all elements must have had large marginals relative to the original ones.

We illustrate Algorithm~\ref{alg:scaling} and the key quantities that appear in it in Figure~\ref{fig:scaling}. The blue line in that figure corresponds to the value $f(T_t)$ for $t = 0, \ldots, |T_0|$. Index $j^\star$ is the largest index for which $f(T_{j^\star}) > \downscaling$, while index $k^\star$ is the smallest index such that $f(T_{k^\star}) \leq (1-\delta) f(T_{j^\star})$. Index $t^\star \in \{j^\star+1, \ldots,k^\star\}$ is an index that maximizes the ratio $\delta_t = \frac{f(i_t | T_t)}{f(i_t \mid T_0 \setminus \{i_t\})}$ among all indices $t$ in that range. The claim is that the set $T_{t^\star-1}$ has the desired properties.

\begin{figure}[t]
\begin{center}
    
\begin{tikzpicture}
\begin{axis}[
    xmin=0, xmax=12,
    ymin=0, ymax=110,
    width=10cm, height=6cm,
    axis lines=left,
    axis line style={->},
    xtick=\empty,
    ytick=\empty,
    grid=major,
    clip=false,  
]
\definecolor{darkgreen}{rgb}{0,0.3,0}

\addplot[mark=o,blue,thick] coordinates {
    (0, 100)
    (1, 90)
    (2, 85)
    (3, 82)
    (4, 65)
    (5, 63)
    (6, 42)
    (7, 33)
    (8, 27)
    (9, 25)
    (10, 12)
    (11, 5)
    (12, 0)
};

\addplot [black, dashed, domain=0:12] {70};

\node at (axis cs: 0,70) [anchor=east, xshift=-5pt, black] {\(\Psi\)};

\addplot [red, dashed, domain=0:12] {82};

\node at (axis cs: 0,82) [anchor=east, xshift=-5pt, black] {\(f(T_{j^\star})\)};

\addplot [red, dashed, domain=0:12] {40};

\node at (axis cs: 0,40) [anchor=east, xshift=-5pt, black] {\( (1-\delta) f(T_{j^\star})\)};

\addplot [darkgreen, dashed] coordinates {(3,0) (3,82)};
\node at (axis cs: 3,0) [anchor=north, yshift=-5pt, black] {\(j^\star \)};
\addplot [darkgreen, dashed] coordinates {(7,0) (7,33)};
\node at (axis cs: 7,0) [anchor=north, yshift=-5pt, black] {\(k^\star\)};

\addplot [darkgreen, dashed] coordinates {(7,0) (7,33)};
\node at (axis cs: 7,0) [anchor=north, yshift=-5pt, black] {\(k^\star\)};

\end{axis}
\end{tikzpicture}
    \caption{Illustration of Algorithm~\ref{alg:scaling}. The blue line corresponds to the value $f(T_t)$ for $t = 0, \ldots, |T_0|$.
    The algorithm first sets a threshold of $\Psi$, and sets $j^\star$ as the last time step with a value larger than $\Psi$. Then it sets a threshold of $(1-\delta)f(T_{j^\star})$, and sets $k^\star$ as the first time step with a value at most this threshold. Between these time steps (excluding $k^\star$), the $f$ value of the set must satisfy Inequality~\eqref{eq:prop_val}, while the proof shows that there must exists a time step within this range with high marginals (in comparison to the marginals of $T_0$) satisfying Inequality~\eqref{eq:prop_marg}.}
    \label{fig:scaling}
    \end{center}
\end{figure}

\begin{proof}[Proof of Lemma~\ref{lem:xosscaling}]
It is easy to see that Algorithm~\ref{alg:scaling} can be implemented in polynomial time with value oracle access to $f$; in particular, $T_0$ can be obtained by iteratively dropping any element from $T$ whose removal does not decrease the value until reaching a point where no such element exists. So we focus on showing that it  finds a set $U$ that satisfies (\ref{eq:prop_val}) and (\ref{eq:prop_marg}).

We begin by arguing that the algorithm is well defined. First note that $f(i \mid T_0 \setminus \{i\}) > 0$ for all $i \in T_0$ because otherwise $T_0$ would not be inclusion-wise minimal. 
Furthermore, we have to argue that the indices $j^\star$ and $k^\star$ as defined in the algorithm exist and satisfy $0 \leq j^\star < k^\star \leq \lvert T_0 \rvert$. 
To this end observe that $f(T) = f(T_0) \geq f(T_1) \geq \ldots \geq f(T_{\lvert T_0 \rvert})$ and $T_{\lvert T_0 \rvert} = \emptyset$, so $f(T_{\lvert T_0 \rvert}) = 0$. Therefore there must be a largest $j$ with $0 \leq j < \lvert T_0 \rvert$ such that $f(T_j)>  \downscaling$. This is $j^\star$.  There must also be a smallest $k$ with $j^* < k \leq \lvert T_0 \rvert$ such that $f(T_k) \leq (1-\delta) f(T_{j^\star})$. This is $k^\star$. 
Finally, since $\delta > 0$, and $f(T_{j^\star})>\downscaling\geq 0$, this $k^\star$ must satisfy $k^\star > j^\star$.

It remains to show that $U = T_{t^\star-1}$ fulfills the respective properties.

By definition, it holds that
\[
f(T_{j^\star})  =  f(T_{k^\star}) + \sum_{t \in \{j^\star+1,\ldots,k^\star\}} f( i_t \mid T_{t} ) .
\]
By rearranging and replacing $f( i_t \mid T_{t} )$ by $ \delta_{t} \cdot f( i_t \mid T_{0} \setminus\{i_t\} )$ we get that
\begin{align*}
f(T_{j^\star})  -  f(T_{k^\star}) &=  \sum_{t \in \{j^\star+1,\ldots,k^\star\}}\delta_{t} \cdot f( i_t \mid T_{0}\setminus\{i_t\} ) \\
& \leq \sum_{t \in \{j^\star+1,\ldots,k^\star\}}\delta_{t^\star} \cdot f( i_t \mid T_{0}\setminus\{i_t\} ) \\
&\leq \delta_{t^\star} \cdot f(\{i_{j^\star+1},\ldots,i_{k^\star}\}) \\ 
&\leq \delta_{t^\star} \cdot f(\{i_{j^\star+1},\ldots,i_{\lvert T_0 \rvert}\}) =\delta_{t^\star} \cdot  f(T_{j^\star}),
\end{align*}
where the first inequality is by the definition of $t^\star$, the second inequality is by Lemma~\ref{lem:xos-sum-of-marginals}, and the third inequality is by monotonicity.
Thus,  for all $i \in U =  T_{t^\star-1}$, 
\begin{eqnarray*}
\frac{f(i\mid T_{t^\star-1} \setminus\{i\} )}{f(i\mid T_0 \setminus \{i\})} & \geq & \delta_{t^\star} \geq  \frac{f(T_{j^\star})  -  f(T_{k^\star}) }{ f(T_{j^\star})} 
\\ & = & 1- \frac{ f(T_{k^\star})}{ f(T_{j^\star})} \geq 1- \frac{(1-\delta) f(T_{j^\star})}{ f(T_{j^\star})} = \delta,
 \end{eqnarray*}
where the first inequality holds by definition of $i_{t^\star}$ and $\delta_{t^\star}$. Therefore, in order to establish Equation~\eqref{eq:prop_marg}, it remains to show that $f(i \mid T_0 \setminus \{i\}) \geq f(i \mid T \setminus \{i\})$. Indeed, since $f(T_0) = f(T)$ and $T_0 \subseteq T$, monotonicity of $f$ implies that
\[
f(i \mid T_0 \setminus \{i\}) = f(T_0) - f(T_0 \setminus \{i\}) \geq f(T) - f(T \setminus \{i\}) = f(i \mid T \setminus \{i\}).
\]

To prove Equation~\eqref{eq:prop_val} observe that  $f(T_{t^\star-1}) \geq f(T_{k^\star-1}) \geq (1-\delta) f(T_{j^\star}) \geq (1-\delta) \downscaling$, and that $f(T_{t^\star-1}) \leq f(T_{j^\star}) \leq f(T_{j^\star+1}) +\max_{i\in T} f(\{i\}) \leq \downscaling+\max_{i\in T} f(\{i\})$. 
\end{proof}

\subsection{Computing a Good Set via a Demand Oracle}

We next present a subroutine that formalizes the argument outlined in the introduction: 
if we knew the value of $f(\sstar \cap \actions')$ and every individual agent is negligible, then by running a demand query at appropriately chosen prices and using the scaling property of XOS functions from the previous section 
we can find a contract such that the principal's utility is in $\Omega(f(\sstar \cap \actions'))$.

Algorithm~\ref{alg:xos-contract-estimate} proceeds as follows. It uses an estimate $\optestimate \approx f(\sstar \cap \actions')$ of the target value, and assumes that $\max_{i \in \actions'} f(\{i\}) \ll \optestimate$. It then sets a price $p_i$ for every agent $i \in \actions'$ and computes a demand set $T$ with respect to these prices. The goal is then to invoke Lemma~\ref{lem:stronger-condition}. To this end, the algorithm finds a subset $U \subseteq T$ in which all marginals are large enough compared to $f(U)$. This is possible using Algorithm~\ref{alg:scaling}.

\begin{algorithm}[t]
\caption{Approximately optimal contract given an estimate of $f(\sstar {\cap A'})$}\label{alg:xos-contract-estimate}
\KwData{An XOS function $f:2^\actions \rightarrow \reals_{\geq 0}$, costs $\{c_{i}\}_{i\in \actions}$, an estimate $\optestimate \geq 0$ of $f(S^\star \cap \actions')$}

\KwResult{A set $U \subseteq \actions'$}

            Let $\downscaling=\frac{\optestimate}{32} - \max_{i \in \actions'} f(\{i\})$.

        For every $i\in \actions'$, let $p_{i} = \frac{1}{2} \cdot \sqrt{ c_i \cdot \optestimate }$ ($p_{i} =\infty$ for $i\in \actions\setminus\actions'$)

    Let $T $ be a demand set with prices $p_{i}$ over the set $\actions'$

    \eIf{$0<\downscaling<f(T)$}{

    $U \leftarrow $  the output of Algorithm~\ref{alg:scaling} on $(f, T, \downscaling, \delta= \frac{1}{2})$
    }
    {
    $U\leftarrow \emptyset$ 
    }
    
\Return{U}
        
\end{algorithm}

\begin{lemma}
Algorithm~\ref{alg:xos-contract-estimate} runs in polynomial time with access to a demand and value oracle to $f$. 
If $\optestimate \leq f(\sstar \cap A' )$, then the set $U \subseteq \actions'$ that it returns satisfies 
\[
g(U) \geq \max\left\{\frac{\optestimate}{128}  - \frac{\max_{i \in A'} f(\{i\})}{4},0\right\}. 
\]
\label{lem:gu128}
\end{lemma}

\begin{proof}
Algorithm~\ref{alg:xos-contract-estimate} 
runs in polynomial time with access to a demand and value oracle to $f$, because it issues a single demand query to $f$ and makes at most one call to Algorithm~\ref{alg:scaling} which runs in polynomial-time with value oracle access to $f$.

We next prove the lower bound on $g(U)$, under the assumption that $\optestimate \leq f(\sstar \cap \actions')$.
If $\downscaling\leq 0$ then the algorithm returns the empty set, which satisfies the claim.
So consider the case $\downscaling > 0$. In this case we have
\begin{eqnarray}
f(T) & \geq & f(T) - \sum_{i \in T} p_i  \nonumber \\ 
& \geq & f(\sstar \cap A') - \sum_{i \in \sstar \cap A'} p_i \nonumber \\
& = & f(\sstar \cap A') - \frac{\sqrt{\optestimate}}{2} \cdot \sum_{i \in \sstar \cap A'} \sqrt{c_i} \nonumber \\
&\geq &  f(\sstar \cap \actions' ) - \frac{1}{2} \sqrt{f(\sstar \cap \actions' ) \cdot  \optestimate}\nonumber \\ 
& \geq &  \frac{1}{2} \cdot f(\sstar \cap \actions') , \label{eq:ft}
\end{eqnarray}
where the first inequality holds since the prices are positive, the second inequality uses that $T$ is a demand set over the set $A'$, the third inequality is by Lemma~\ref{lem:stronger-sumci}, and the last inequality holds because $\optestimate \leq f(\sstar \cap \actions')$.

Furthermore, note that, as $T$ is a demand set, we have to have $f(i \mid T \setminus\{i\}) \geq p_i$ for all $i \in T$ (because otherwise it would be beneficial to remove $i$).

Now, since $\downscaling = \optestimate/32 - \max_{i \in \actions'} f(\{i\}) > 0$, we must have $\optestimate > 0$ and hence $\optestimate/2 > \optestimate/32 \geq \downscaling$. This shows
$$f(T)\stackrel{\eqref{eq:ft}}{\geq} \frac{1}{2}\cdot f(\sstar \cap \actions') \geq  \frac{\optestimate}{2} 
> \downscaling.$$
By Lemma~\ref{lem:xosscaling}, 
the set $U \subseteq T$ therefore fulfills $$f(U) \geq \left(1-\delta\right) \downscaling = 
\frac{1}{64} \optestimate - \frac{1}{2} \max_{i \in \actions'} f(\{i\}).$$
Furthermore, $U$ fulfills 
\begin{equation}
    f(U) \leq \downscaling + \max_{i \in T} f(\{i\}) \leq  \frac{\optestimate}{32}\label{eq:fu-up}
\end{equation} 
and  for all $i \in U$
\begin{equation}
f(i \mid U \setminus\{i\}) \geq \delta f(i \mid T \setminus\{i\}) = \frac{1}{2} \cdot  f(i \mid T \setminus\{i\}) \label{eq:fu-marg}.    
\end{equation}
We conclude that, for all $i \in U$, we have
\begin{eqnarray*}
    f(i \mid U \setminus\{i\}) 
    & \stackrel{\eqref{eq:fu-marg}}{\geq} & \frac{1}{2} \cdot f(i \mid T \setminus\{i\}) 
    \geq \frac{1}{2}\cdot  p_i  \\ & =  &\frac{1}{4}\cdot \sqrt{c_i \optestimate} = \sqrt{2 c_i \frac{\optestimate}{32}} \stackrel{\eqref{eq:fu-up}}{\geq} \sqrt{2 c_i f(U)}.
\end{eqnarray*}
So, Lemma~\ref{lem:stronger-condition} implies $g(U) \geq \frac{1}{2} f(U)$. In combination, we obtain
\[
g(U) \geq \frac{1}{2}f(U) \geq \frac{1}{2} \left( \frac{1}{64} \optestimate - \frac{1}{2} \max_{i \in \actions'} f(\{i\}) \right),
\]
as claimed.
\end{proof}

\subsection{Putting It All Together}

We are finally in a position to wrap up our argument. Below we show that Algorithm~\ref{alg:xos-contract} achieves a constant-factor approximation to $g(\sstar)$ for XOS reward function $f$, and runs in polynomial time with access to a demand and a value oracle to $f$. 

The idea is as follows: Try out all contracts that incentivize a single agent; and use the subroutine from the previous section with polynomially many guesses for the value of $f(\sstar \cap \actions')$ to obtain additional candidate sets. 
Argue that for the right guess of $f(\sstar \cap \actions')$ the better of the best contract that incentivizes a single agent, and the candidate set $U$ from the subroutine for this guess  yields the desired approximation.

In more detail, Algorithm~\ref{alg:xos-contract} defines a family $\cand$ of candidate sets of agents to be incentivized. First, it adds the empty set and all singleton sets. Then it adds additional candidate sets using Algorithm~\ref{alg:xos-contract-estimate}.
According to Lemma~\ref{lem:gu128}, one would ideally set $\optestimate = f(\sstar \cap A' )$ in 
Algorithm~\ref{alg:xos-contract-estimate}. However, as we do not know $f(\sstar \cap A' )$, Algorithm~\ref{alg:xos-contract} instead tries different values of $\optestimate$, namely all powers of $\xi$.

\begin{algorithm}[t]
\caption{Approximate optimal contract for XOS}\label{alg:xos-contract}
\KwData{An XOS function $f:2^\actions \rightarrow \reals_{\geq 0}$, costs $\{c_{i}\}_{i\in \actions}$, parameter $\xi>1$}
\KwResult{Set $S$ to incentivize}

    Let $\cand=\{\{i\} \mid i\in \actions \} \cup \{\emptyset\}$

    Let $\actions' =  \{i \in \actions \mid \frac{c_i}{f(\{i\})} \leq \frac{1}{2}\} $

\If{$\actions'\neq \emptyset \wedge \max_{i\in \actions'}f(\{i\}) >0$}{
\For{$ j =0,.. ,\lceil \log_{\xi} 2n \rceil $}
{
    Let $\optestimateround{j} = \xi ^ j \cdot \max_{i\in \actions'} f(\{i\})/2$

    $U^{(j)} \leftarrow $ the output of Algorithm~\ref{alg:xos-contract-estimate} on ($f,\{c_i\}_{i\in \actions},\actions',\optestimate = \optestimateround{j}$)  
    
    $\cand \leftarrow \cand \cup \{U^{(j)}\}$
}
}
\Return{$\arg\max_{S \in \cand} g(S)$}
    
\end{algorithm}

\begin{proof}[Proof of Theorem~\ref{thm:xos}]
We show that Algorithm~\ref{alg:xos-contract} achieves the claims in the theorem for XOS reward functions $f$, assuming that we have access to a demand and a value oracle to $f$. We present details on how to obtain the result for submodular reward functions $f$ when we only have value oracle access to $f$ in Appendix~\ref{app:extension-to-submodular}.

We begin by observing that, for any $\xi > 1$, Algorithm~\ref{alg:xos-contract} runs in polynomial time with access to a demand and a value oracle to $f$, because we make only polynomially many calls to Algorithm~\ref{alg:xos-contract-estimate}, which itself runs in polynomial time with demand/value oracle access to $f$.

We next show that Algorithm~\ref{alg:xos-contract} obtains an approximation guarantee of $1/(256\xi+2)$ which tends to $1/258$ as $\xi \rightarrow 1$.
To this end, we distinguish between two cases. The first case is that  $f(\sstar\cap \actions') \leq 128 \xi \cdot \max\{0,\max_{i\in \actions} g(\{i\})\}$. Then by Lemma~\ref{lem:onebigagent}
\begin{align*}
g(\sstar) & \leq f(\sstar \cap A') + \max\{0,\max_{i \in A} g(\{i\})\} \\ & \leq (128 \xi + 1) \cdot \max\{0,\max_{i \in A} g(\{i\})\}.
\end{align*}
So, the best single action (or the empty set) gives at least an approximation of $\frac{1}{128\xi+1}$, which completes the proof of the theorem for this case.

In the other case, we have $$f(\sstar\cap \actions') > 128 \xi \cdot \max\{0,\max_{i\in \actions} g(\{i\})\}.$$ So, in particular, we need to have $A' \neq \emptyset$ and $\max_{i \in A'} f(\{i\}) > 0$. Furthermore, $g(\{i\}) \geq \frac{1}{2} f(\{i\})$ for all $i \in A'$, resulting in $\max_{i\in \actions} g(\{i\}) \geq \max_{i\in \actions'} f(\{i\})/2$. So, in combination in this case $f(\sstar\cap \actions') \geq \max_{i\in \actions'} f(\{i\})/2$. On the other hand, by subadditivity of $f$,
$f(\sstar \cap \actions') \leq \sum_{i \in \sstar \cap \actions'} f(\{i\}) \leq 2n \cdot \max_{i\in \actions'} f(\{i\})/2$. Therefore, there is a unique $j^\star \in \{0, 1, \ldots, \lceil\log_{\xi} 2n\rceil\}$ for which $\optestimateround{j^\star} \leq f(\sstar \cap \actions') < \xi \cdot \optestimateround{j^\star}$.

Now, since $\frac{1}{\xi} \cdot f(\sstar \cap \actions') \leq \optestimateround{j^\star} \leq f(\sstar \cap \actions')$,
\begin{eqnarray*}
g(U^{(j^\star)})    & \geq &  \frac{\optestimateround{j^\star}}{128} - \frac{\max_{i\in \actions'} f(\{i\})}{4}  \\   
&\geq & \frac{f(\sstar\cap \actions')}{128\cdot \xi} - \frac{\max_{i\in \actions} g(\{i\})}{2} \\ & \geq &  f(\sstar\cap \actions')\cdot \left(\frac{1}{128\cdot \xi} - \frac{1}{256\cdot \xi} \right) \\  & \geq &
\frac{128\xi \cdot g(\sstar)}{128\xi+1}\frac{1}{256\cdot \xi} = \frac{g(\sstar)}{256\xi+2},
\end{eqnarray*}
where the first inequality is by Lemma~\ref{lem:gu128}, the second inequality is since for every $i\in \actions'$ it holds that $\frac{c_i}{f(\{i\})}\leq \frac{1}{2}$ and thus $g(\{i\}) = f(\{i\})(1-\frac{c_i}{f(\{i\})}) \geq f(\{i\})/2$, the third and fourth inequality both use that $f(\sstar\cap \actions') > 128 \xi \cdot \max\{0,\max_{i\in \actions} g(\{i\})\}$, and the the fourth inequality additionally uses Lemma~\ref{lem:onebigagent}.
This  gives an approximation of $\frac{1}{256\xi+2}$.
\end{proof}

\section{Inapproximability Results for Submodular and Subadditive}
\label{sec:beyond-submodular}

In this section, we present our 
inapproximability results. We establish a constant lower bound for submodular reward functions (Section~\ref{sec:submodular-upper}), and a $\Omega(\sqrt{n})$ lower bound for subadditive reward functions (Section~\ref{sec:subadd-lower}). 
Both our impossibility results are for algorithms that have access to value and demand queries. 
A striking feature of the impossibility result for subadditive reward functions is that the subadditive reward functions used in the proof are constant-factor close to submodular  
(see Observation~\ref{obs:subadditive-close}).

\subsection{Inapproximability Result for Submodular}\label{sec:submodular-upper}

We start with the inapproximability result for submodular reward functions.
We show that there is some constant $\eta > 1$ 
for which no algorithm (deterministic or randomized) can guarantee an approximation of $\eta$ with a polynomial number of value and demand queries. 
This result strengthens a result by \citet{EzraFS24}, who show  that no polynomial-time algorithm can get a 
$1.42$-approximation with access to a value oracle alone (unless P $=$ NP).

\begin{theorem}\label{thm:upper}
    There exists a constant 
    $\eta > 1$
    for which any (possibly randomized)
    algorithm that uses a sub-exponential number of value and demand queries to the submodular 
    reward function returns a $\eta$-approximate optimal contract with an exponentially small probability in $n$.  
\end{theorem}

To prove the theorem, we define a probability distribution over instances and consider an arbitrary deterministic algorithm on a randomly drawn instance. By Yao's principle, the worst-case approximation ratio of any randomized algorithm is lower-bounded by the approximation ratio of any deterministic algorithm on this randomized instance.

Our proof relies on the construction of a family of functions $\{f_T:2^\actions\rightarrow [0,1]\}_T$, defined using the following function over two variables $x,y$ and a parameter $k$:
\[
f(x,y; k) = \begin{cases}
    \sqrt{k} & \text{ if $x+y \geq k$} \\
    \sqrt{x+y} & \text{ if $x+y < k$ and $y \leq \frac{k}{2}$} \\
    \sqrt{x + \frac{k}{2}}  +  \frac{y- \frac{k}{2}}{\sqrt{k} + \sqrt{x + \frac{k}{2}}} & \text{ otherwise.}
\end{cases}
\]
The function $f_T$ is given by 
$$ f_T(S)= f(|S\setminus T|, |S\cap T|; k)/\sqrt{k} \quad\text{ for $k=|T|$.}$$

Let $g_{T}$ be the corresponding principal's utility function with respect to $f_{T}$, when the cost of each agent $i$ is
$c_i = \frac{1}{8k^{2}}$.

The idea behind this construction is that all sets of size $k$ have the same $f$-value; however for such a set to achieve high principal utility it should also have high marginals. This is achieved by the different behavior of $f_T$ on smaller sets, depending on the size of the intersection with $T$. For sets admitting high intersection with $T$, 
the function $f$ is growing linearly on subsets as opposed to the (concave) square-root function everywhere else.

The key ingredients of our proof are:

\begin{enumerate}
\item In Claim~\ref{cl:submodularity} we show that for every choice of $T\subset A$ of even size, the reward function $f_T$ is monotone and submodular. 
\item In Lemma~\ref{lem:sub-opt} we show that any set $S$ that gives a good approximation to $\max_{X'} g_{T}(X')$ must have size at most $k$ and a large intersection with $T$ (specifically, $|S \cap T| > k/2$).
\item In Lemma~\ref{lem:demand} we show that for any $T \subseteq A$ of even size, a demand query for $f_T$ can be answered with polynomially many value queries.
\end{enumerate}

To complete the proof, we choose a uniformly random set $T$ of size $k = 2 \cdot \lfloor \frac{n}{10} \rfloor$. Since $k$ is of even size, $f_T$ is monotone submodular by (1). 
To show hardness under value queries, we use (2) to argue that, 
for any set of size $j \in [k/2+1,k]$, among the $\binom{n}{j}$ sets of size $j$, only $\binom{k}{k/2} \cdot \binom{n}{j-k/2}$ can satisfy the required approximation, and this number is exponentially smaller than $\binom{n}{j}$ for every such $j$.
Finally, by (3), hardness under demand queries reduces to hardness under value queries. 

We proceed with the formal proof.

Claim~\ref{cl:submodularity} is deferred to the appendix. 
We next present the lemma establishing (2), namely that in order to achieve a good approximation, we need to incentivize a set of size at most $k$ with a large intersection with $T$.

\begin{lemma}\label{lem:sub-opt}
     We have $\max_{X'}g_T(X') > 0.78 $. For any set $X$ with $\lvert X \rvert \leq k$ and $\lvert X \cap T \rvert \leq \frac{k}{2}$, we have $g_T(X) < (0.97 + \frac{1}{3 \sqrt{k}}) \cdot \max_{X'} g_T(X')$.
     For any set $X$ with $|X| > k$, we have $g_T(X) = - \infty$.
\end{lemma}

\begin{proof}
We first show that $\max_{X'} g_T(X') > 0.78 $. 
In particular, we show that $g_T(T) > 0.78$.
First observe that $f_T(T) = 1$. Furthermore, for every $i \in T$, we have
\[
f_T(T \setminus \{i\}) = \sqrt{\frac{1}{2}} + \left(1 - \sqrt{\frac{1}{2}}\right) \frac{k - 1 - \frac{k}{2}}{\frac{k}{2}} = \sqrt{\frac{1}{2}} + \left(1 - \frac{1}{\sqrt{2}}\right) \left(1 - \frac{2}{k}\right).
\]
So the marginal of every element is given by
\[
f_T(i \mid T \setminus \{i\}) =  f_T(T) - f_T(T \setminus \{i\}) = \left(1 - \frac{1}{\sqrt{2}}\right) \frac{2}{k}  = \frac{2-\sqrt{2}}{k},
\]
implying 
\[
\frac{c_i}{f_T(i \mid T \setminus \{i\})} = \frac{1}{8 k^2} \cdot \frac{k}{2 - \sqrt{2}} = \frac{1}{k} \cdot \frac{1}{8(2 - \sqrt{2})}
\]
and so
\[
1 - \sum_{i \in T} \frac{c_i}{f_T(i \mid T \setminus \{i\})} = 1 - \frac{1}{8(2 - \sqrt{2})} > 0.78.
\]
This implies that
\[
g_T(T) = \left( 1 - \sum_{i \in T} \frac{c_i}{f_T(i \mid T \setminus \{i\})}\right) f_T(T) >0.78 .
\]

Given that $g_T(T) > 0.78$, in order to prove the claimed upper bound on $g_T(X)$ for $X$ such that $\lvert X \rvert \leq k$ and $\lvert X \cap T \rvert \leq \frac{k}{2}$, it suffices to show that for any such $X$ we have $g_T(X) \leq \frac{3}{4}  + \frac{1}{4\sqrt{k}}$. 

To this end, observe that, for any such $X$ and all $i \in X$, we have

\begin{align*}
\frac{c_i}{f_T(i \mid X \setminus \{i\})} = \frac{c_i}{(\sqrt{\lvert X \rvert} - \sqrt{\lvert X \rvert - 1})/\sqrt{k}} &= \sqrt{k}c_i (\sqrt{\lvert X \rvert} + \sqrt{\lvert X \rvert - 1}) \\& \geq 2\sqrt{k} c_i \sqrt{\lvert X \rvert - 1} \geq 2 \sqrt{k}c_i (\sqrt{\lvert X \rvert} - 1).
\end{align*}
Therefore,
\begin{align*}
g_T(X) 
&= \left( 1 - \sum_{i \in X} \frac{c_i}{f_T(i \mid X \setminus \{i\})}\right) f_T(X) \leq \left(1 - \lvert X \rvert \cdot 2\sqrt{k} c_i (\sqrt{\lvert X \rvert} - 1)\right) \sqrt{\lvert X \rvert/k} \\
&= \sqrt{\lvert X \rvert/k} - \frac{\lvert X \rvert^2}{4 k^{2}} + \frac{\lvert X \rvert^{3/2}}{4 k^{2}} \leq \frac{3}{4}  + \frac{1}{4\sqrt{k}},
\end{align*}
as desired, where we use that $x \mapsto \sqrt{x/k} - \frac{x^2}{4 k^{2}}$ is non-decreasing for $x \in [0, k]$.

The proof is completed by noting that for any $X$ such that $\lvert X \rvert > k$, we have $f_T(i \mid X \setminus \{i\}) = 0$ and thus $g_T(X) = -\infty$.
\end{proof}

We next present the lemma establishing (3), namely, how to implement a demand query for $f_T$ using poly-many value queries to $f_T$.
\begin{lemma}
\label{lem:demand}
    For every set $T \subseteq \actions $ of even size $4 \leq k<|\actions|$,  
    answering a demand query for $f_T$ can be done with polynomially many value queries to $f_T$ (without knowing the set $T$).
\end{lemma}

\begin{proof}
We present an algorithm that answers a demand set using poly many value queries. 

Let $p_1,\ldots,p_n$ be the 
price vector whose demand is desired; that is, 
we want to find a set $S$ that maximizes  $f_T(S) -\sum_{j\in S} p_j$. Without loss of generality assume that $p_1 \leq p_2\leq \ldots \leq p_n$.

We first observe that the algorithm can learn $k$ (i.e., the size of $T$), by a single value query because $f_T(\{i\}) = 1 / \sqrt{k}$ for every $i \in \actions$. 

Next, the algorithm issues value queries of the form $f_T([i])$ for all $i\in [k]$.
If $f_T([i]) = \sqrt{i/k}$ for all $i\in [k]$, then the algorithm returns the set $S$ that maximizes $f_T(S) -\sum_{j\in S} p_j$ among all these sets. This must be a demand set since $[i]$ is the cheapest set of size $i$, no set of size $i$ has value greater than $\sqrt{i}$, and there is always a demand set of size at most $k$.

Otherwise, there exists $i \in [k]$ such that $f_T([i]) \neq \sqrt{i/k}$. 
We show that in this case the algorithm can find the set $T$ with poly-many additional value queries. We then show that if the algorithm knows $T$ then it can find a demand set with poly-many value queries.

Let $i^\star$ be the minimal $i$ such that $f_T([i]) \neq \sqrt{i/k}$. By the definition of $f_T$, we get that $i^\star<k$ and $|[i^\star] \cap T |> \frac{k}{2}$.
Since $f_T([i^\star-1]) =\sqrt{(i^\star-1)/k}$ 
and $i^\star-1 <k$, it must be that $|[i^\star-1] \cap T| \leq \frac{k}{2} $, which implies that $i^\star \in T$.

In order to determine which of the actions $i^\star+1, \ldots, n$ belong to $T$, the algorithm issues for each $j \notin [i^\star]$ the query $f_T([i^\star-1]\cup\{j\})$. 
If for $j$ it holds that $ f_T([i^\star-1]\cup\{j\}) = f_T([i^\star]) $, it means that 
$j\in T$, otherwise, $f_T([i^\star-1]\cup\{j\})  = \sqrt{i^\star/k}$ and $j \notin T$.
Finally, in order to determine whether $j \in T$ for any $j \in [i^\star-1]$,  
let $j^\star$ be an arbitrary action in $[n] \setminus [i^\star]$ 
(such exists since $i^\star <k$).  The algorithm issues the query  $ f_T([i^\star] \cup \{j^\star\} \setminus \{j\})  $. If $ f_T([i^\star] \cup \{j^\star\} \setminus \{j\})  = f_T([i^\star])$ then $j \in T $ iff $j^\star \in T$. Otherwise, $j \in T $ iff $j^\star \not\in T$. Since whether $j^\star$ in $T$ is already known, it means that we know for each $i\in \actions$ whether it is in $T$ or not.

Once the identity of $T$ is known, one can answer a demand query as follows.
Go over all options for $x,y$ satisfying $0 \leq x+y \leq k$, and consider the set $S$ composed of the $x$ cheapest elements in $\actions \setminus T$, and the $y$ cheapest elements in $T$. The set $S$ among these sets that maximizes $f_T(S) -\sum_{j\in S} p_j$ is a demand set.

This algorithm returns a demand set in polynomial time using $O(|\actions|)$ value queries, completing the proof of the lemma.
\end{proof}

We are now ready to present the proof of Theorem~\ref{thm:upper}.

\begin{proof}[Proof of Theorem~\ref{thm:upper}] 
To prove the theorem, it suffices to show that the statement holds with respect to a deterministic algorithm against a randomized input, where the set $T$ that defines the function $f_T$ is chosen uniformly at random from all sets of size $k = 2 \cdot  \lfloor \frac{n}{10}\rfloor$. 
(Note that restricting the sets $T$ that the adversary chooses from only strengthens our statement.)

By Lemma~\ref{lem:demand}, it suffices to consider a deterministic algorithm that issues only value queries, and we may assume without loss of generality that the returned set is queried (otherwise, just add this query at the end of the algorithm).

By Lemma~\ref{lem:sub-opt}, for large enough $k = 2 \cdot \lfloor \frac{n}{10} \rfloor$, 
in order to get better than $\eta = 1.03$ approximation, 
the algorithm must issue a value query for a set $S$ for which  $|S| \leq k $ and $|S \cap T| > \frac{k}{2}$. We call such a query a {\em good} query.

Since the probability of a query $S$ to be good is increasing in the size of the set $S$ (up to size $k$), the probability that a single query is successful is upper bounded by the probability that a hypergeometric random variable with parameters $(n,k,k)$ is greater than $k/2$. For our choice of $k$ (namely, $k = 2 \cdot  \lfloor \frac{n}{10}\rfloor$), this probability is exponentially small in $n$.

We note that, since we bound the probability that the algorithm issues a good query, it is without loss of generality to assume that the queries are non-adaptive. Indeed, at any point in time, it can be assumed that none of the previous queries were good, and for queries that are not good, no information beyond this fact is learned.

In summary, every query is good with exponentially small probability, thus, by the union bound, the probability of issuing a good query is exponentially small for a sub-exponential number of queries.
This concludes the proof.
\end{proof}

We remark that our proof applies to submodular success probability functions, but does not rule out the existence of a PTAS (or even FPTAS) for the special case of gross-substitutes functions---an important special case of submodular functions. The definition of gross-substitutes functions, along with the observation and proof that our construction does not satisfy the gross substitutes condition appear in Appendix~\ref{app:upper-submod}.

\subsection{Inapproximability Result for Subadditive}
\label{sec:subadd-lower}

In the remainder of this section, we state and prove our inapproximability result for subadditive reward functions. We show that for this broader class of reward functions, no algorithm can achieve a better than $O(\sqrt{n})$ approximation.

\begin{theorem}\label{thm:subadditive-demand}
For every randomized algorithm that performs only polynomially many value and demand queries, there is a subadditive function $f$ such that the approximation ratio is $\Omega(\sqrt{n})$ with high probability.
\end{theorem}

For simplicity of the proof, we use a non-normalized reward function $f:2^\actions \rightarrow \mathbb{R}_{\geq 0}$ instead of a normalized reward function $f:2^\actions \rightarrow [0,1]$. The proof can be adjusted to use a normalized function by dividing both the function and all the costs by $f(A)$.

As in our impossibility result for submodular rewards, we use that, in order to show an impossibility for randomized algorithms, it suffices to define a probability distribution over instances and consider an arbitrary deterministic algorithm on a randomly drawn instance.

To this end, we construct a subadditive function $f$ such that $g(S) = O(1)$ for $|S| < \sqrt{n}$ and $g(S) \leq 0$ if $|S| \geq \sqrt{n}$. We then modify $f$ by increasing the value of a random set $T$ of size $\frac{n}{2}+1$. 
We argue that monotonicity and subadditivity are preserved by this change. We then show that this change significantly increases the principal's utility from $T$ to $\Omega(\sqrt{n})$, while keeping the principal's utility from any other set $S \neq T$ at $O(1)$. Hence, in order to achieve a better than $O(\sqrt{n})$-approximation to the optimal principal's utility, the algorithm has to identify the randomly chosen set $T$. Finally, we show that any value or demand query reveals information only about a polynomial number of candidate sets. Since there are exponentially many choices of the target set $T$, this leads to the claimed inapproximability result.

We describe the random instances in Section~\ref{sec:subadditive-instances}, and establish properties of the reward functions and the principal's utility function used in the construction in Section~\ref{sec:subadditive-reward} and Section~\ref{sec:subadditive-principal-utility}. 
We prove a lemma that limits the power of demand queries in these random instances in Section~\ref{sec:subadditive-demand-queries}, and give the proof of Theorem~\ref{thm:subadditive-demand} in Section~\ref{sec:subadditive-proof}.

\subsubsection{Distribution Over Subadditive Reward Functions \texorpdfstring{$f_T$}{f\_T}}\label{sec:subadditive-instances}

To define the probability distribution over instances, let $n \in \mathbb{N}$ be an even square number. In every instance, the agents' costs are the same, namely $c_i=\frac{2}{n}$ for every agent $i$.

The instances differ in terms of the reward functions. To define them, consider the following 
(monotone) subadditive function $f$ over $n$ agents:

\[
f(S)=
\begin{cases}
3+\frac{2 \cdot |S|}{\sqrt{n}} & |S|\leq \frac{n}{2}\\
4  + \sqrt{n} & |S|= \frac{n}{2}+1 \\
5 + \sqrt{n} & |S| = \frac{n}{2}+2 \\
6 + \sqrt{n} & |S| \geq \frac{n}{2}+3 \\
\end{cases}
\]

Given a set $T$ of size $n/2+1$, let $f_T$ be the function where $f_T(S) = f(S) + \indicator{S=T}$.  Accordingly, let $g_T$ be the associated principal's utility function. 
The random instance consists of $f_{T^\star}$, where $T^\star$ is a uniformly drawn set of size $n/2 + 1$.

\subsubsection{Properties of the Reward Function \texorpdfstring{$f_T$}{f\_T}}\label{sec:subadditive-reward}

We state and prove properties of the reward functions $f_T$.
We first show that, for every $T$, the function $f_T$ is monotone and subadditive.

\begin{claim} \label{cl:f-sub}
For every $T$, the function $f_T$ is monotone and subadditive.
\end{claim}

\begin{proof}
We first argue that $f_T$ is monotone. Since $f$ is monotone and $f_T$ only increases the value of the set $T$ by $1$, it suffices to observe that for any $S \supset T$ it holds that $f_T(S) = f(S) \geq f_T(T)$. To establish subadditivity, consider any disjoint sets $S_1,S_2 \subseteq \actions$. We show that $f_T(S_1)+f_T(S_2) \geq f_T(S_1 \cup S_2)$. Distinguish two cases. 

\noindent {\bf Case 1: $|S_1|,|S_2|\leq \frac{n}{2}$.}
In this case $f_T(S_1)+f_T(S_2) = 6+\frac{2\cdot (|S_1|+ |S_2|)}{\sqrt{n}}$. 
If $|S_1|+|S_2|\leq \frac{n}{2}$, then $f_T(S_1\cup S_2) \leq 3+\frac{2\cdot(|S_1|+|S_2|)}{\sqrt{n}} \leq 6+\frac{2\cdot(|S_1|+|S_2|)}{\sqrt{n}} =f_T(S_1)+f_T(S_2)$.
If $|S_1|+|S_2| > \frac{n}{2}$, then $f_T(S_1\cup S_2) \leq 6+\sqrt{n} \leq 6+\frac{2\cdot(|S_1|+|S_2|)}{\sqrt{n}}=f_T(S_1)+f_T(S_2)$.

\noindent {\bf Case 2: $|S_1|>\frac{n}{2}, |S_2|\geq 1$.}
In this case $f_T(S_1)+f_T(S_2) \geq 4+\sqrt{n}+3 \geq f_T(S_1 \cup S_2)$.
\end{proof}

We also observe that for every $T$, the function $f_T$ is constant-factor close to a submodular function $f'$. In fact, it is even constant-factor close to a \emph{symmetric} submodular function $f'$, i.e., a submodular function that depends only on the cardinality of the set.\footnote{Since symmetric submodular functions are gross-substitutes \citep{gul1999walrasian}, these subadditive functions are even close to gross-substitutes functions.}

\begin{observation}\label{obs:subadditive-close}
There is a (monotone) symmetric submodular function $f'$ that fulfills
\[
f'(S) \leq f_T(S) \leq \left( 1 + \frac{3}{3 + \sqrt{n}} \right) f'(S)
\]
for all $S$ and $T$.
\end{observation}

\begin{proof}
Let $f'(S) = \min\{ 3 + \frac{2 \cdot \lvert S \rvert}{\sqrt{n}}, 3 + \sqrt{n} \}$. Note that $f'$, as the minimum between two submodular functions, is submodular. Observe that $f'(S) \leq f_T(S)$ for all $S \subseteq \agents$ with equality for $\lvert S \rvert \leq \frac{n}{2}$. For $\lvert S \rvert > \frac{n}{2}$, we have $f'(S) = 3 + \sqrt{n}$ and $f_T(S) \leq 6 + \sqrt{n}$.
\end{proof}

\subsubsection{Properties of the Principal's Utility Function \texorpdfstring{$g_T$}{g\_T}}\label{sec:subadditive-principal-utility}

We next show that $T$ is the optimal contract for reward function $f_T$, and no other set gives a better than $O(\sqrt{n})$-approximation to 
$g_T(T)$.
This implies that one needs to identify $T$ in order to get a better than $O(\sqrt{n})$-approximation.

\begin{lemma}\label{lem:onlyT}
We have that $g_T(T) \geq \frac{\sqrt{n}}{4}$ and $g_T(S) \leq 5$ for all $S \neq T$.
\end{lemma}

\begin{proof}
First note that for every $i \in T$, $f_T(i \mid T \setminus \{i\})=2$.
Thus, $$g_T(T)=f_T(T)\left(1-\sum_{i\in T}\frac{c_i}{f_T(i \mid T \setminus \{i\})}\right) \geq \frac{\sqrt{n}}{4}.$$ 
Next, consider any $S \neq T$. We can distinguish the following cases. 
\begin{itemize}
    \item If $|S|\leq 1$, then $g_T(S)\leq f_T(S) \leq  3+\frac{2}{\sqrt{n}}$.
    \item If $1 < |S| \leq \frac{n}{2}$, then $f_T(i \mid S \setminus \{i\}) = \frac{2}{\sqrt{n}}$, thus 
    \begin{eqnarray*}
    g_T(S) & = & f_T(S)\left(1-\sum_{i\in S}\frac{c_i}{f_T(i \mid S \setminus \{i\})}\right) = f_T(S)\left(1-\frac{1}{\sqrt{n}}|S|\right) \\ & \leq & \begin{cases}
    5\cdot \left(1-\frac{|S|}{\sqrt{n}}\right) & |S|<\sqrt{n}\\
    0 & |S|\geq \sqrt{n}\\
    \end{cases}
    \end{eqnarray*}
    \item If $|S| \geq \frac{n}{2}+1$ and $S\neq T$, then 
    $f_T(i \mid S \setminus \{i\}) \leq 1$ for all $i\in S$, thus 
    $$g_T(S)=f_T(S)\left(1-\sum_{i\in S}\frac{c_i}{f_T(i \mid S \setminus \{i\})}\right) \leq f_T(S)\left(1-|S|\cdot \frac{2}{n}\right) \leq 0.$$
\end{itemize}
Hence for any $S \neq T$ it holds that $g_T(S) \leq 5$.
\end{proof}

\subsubsection{Limiting the Power of Demand Queries}\label{sec:subadditive-demand-queries}

We now show that with each demand query one can distinguish at most $n^{29}$ sets $T^\star$. 
Together with Lemma~\ref{lem:onlyT}, this will yield our impossibility result.

\begin{lemma} \label{lem:size-dem}
For $n>4096$, for every vector of prices $p=(p_1,\ldots,p_n)$, 
the set of $\{T^\star \mid D(f,p) \neq D(f_{T^\star},p)\}$ is at most of size $n^{29}$, where $D(f,p)$  (resp. $D(f_{T^\star},p)$) is the demand set\footnote{For simplicity of the proof, we assume that when there are multiple sets  in demand, the tie breaking is consistent across all $f_T$.} of function $f$ (resp. $f_{T^\star}$) with respect to prices $p$. 
\end{lemma}
Let $S_p=\{i \mid p_i \leq 1/4\}$, and $B_p=\{i \mid p_i \geq 1/2\}$.

\begin{claim}
If $|B_p| \leq n/2 -3$, then for every $T$ {it holds that} $D(f_T,p)=D(f,p)$. \label{cl:big}
\end{claim}
\begin{proof}
Since $f_T$ and $f$ disagree only on the value of $T$, it is sufficient to show that $D(f_T,p) \neq T$.
Assume towards contradiction that $D(f_T,p) = T$, then, {because $|A \setminus B_p| \geq n/2+3$ and $|T| = n/2+1$}, the set $\Delta:= \actions \setminus B_p \setminus T $ is of size at least $2$. Let $x,y$ be two arbitrary different elements in $\Delta$, {and note that for these $p(x) < 1/2$ and $p(y) < 1/2$}.
Now $f_T(T\cup \{x,y\}) - p(T\cup \{x,y\}) \geq 1+ f_T(T) -p(T) -p(x)-p(y) > f(T)-p(T)$ which contradicts that $T$ is the demand set for $f_T$.
\end{proof}

\begin{claim}
For $n>4096$, if $|S_p| \leq n/2 - 8$, then for every $T$ {it holds that} $D(f_T,p)=D(f,p)$. \label{cl:small}
\end{claim}
\begin{proof}
Since $f_T$ and $f$ disagree only on the values of $T$, it is sufficient to show that $D(f_T,p) \neq T$.
Assume towards contradiction that $D(f_T,p) = T$, then the set $\Delta:= T \setminus S_p  $ is of size at least $9$. Let $X$ be  an arbitrary subset of  $\Delta$ of size $9$. It holds that 
\begin{eqnarray*}
f_T(T\setminus X)-p(T \setminus X) & = & f_T(T)-2-\frac{{16}}{\sqrt{n}} - p(T) + p(X) \\ &\geq & f_T(T)-p(T) +\frac{|X|}{{4}} -2-\frac{16}{\sqrt{n}}  \\ &> & f_T(T)-p(T),
\end{eqnarray*}
which contradicts that $T$ is the demand set.
\end{proof}

\begin{claim}
If  $|S_p| > n/2 - 8$ and  $|B_p| > n/2 -3$, then for all $T \notin \{S \mid (|S| =n/2+1) \wedge (|S\setminus S_p| \leq 14) \}$ {it holds that} $D(f_T,p)=D(f,p)$. \label{cl:mixed}
\end{claim}
\begin{proof}
Since $f_T$ and $f$ disagree only on the values of $T$, it is sufficient to show that $D(f_T,p) \neq T$.
Assume towards contradiction that $D(f_T,p) = T$,
it holds that the set $\Delta_1:= B_p \cap T $ is of size at least $5$ since $|B_p \cap T| \geq |T\setminus S_p| -|\actions \setminus B_p \setminus S_p| \geq 14-9 =5 $, and the set $\Delta_{{2}}:= S_p \setminus T $ is of size at least $6$ since $|S_p \setminus T| = |T \setminus S_p | -  |T| +|S_p| \geq 14 - 8  =6.$
Let $X$ be  an arbitrary subset of  $\Delta_1$ of size $5$ and let $Y$ be an arbitrary subset of $\Delta_2$ of size $5$. It holds that
\begin{eqnarray*}
& & f_T((T\setminus X) \cup Y)-p((T \setminus X) \cup Y) \\ & = & f_T(T)-1 - p(T) + p(X) -p(Y) \\ & \geq & f_T(T)-p(T) +\frac{|X|}{4} -1  >  f_T(T)-p(T),
\end{eqnarray*}
which contradicts that $T$ is the demand set.
\end{proof}

\begin{proof}[Proof of Lemma~\ref{lem:size-dem}]

By combining Claim~\ref{cl:big} and Claim~\ref{cl:small}, 
a demand query can only reveal information about $T$ if
$|S_p| > n/2 - 8$ and  $|B_p| > n/2 -3$, and even then, by Claim~\ref{cl:mixed}, 
{a demand query}
cannot distinguish between $T$'s not in $\{S \mid (|S| =n/2+1) \wedge (|S\setminus S_p| \leq 14) \}$. 
{Now,} the lemma follows since for every choice of $S_p$ of {size} greater than $n/2-8$, the set $\{S \mid (|S| =n/2+1) \wedge (|S\setminus S_p| \leq 14) \}$ is at most of size $n^{29}$. 
(One can bound it by counting the options to select  a set $S\setminus S_p$ of size at most $14$ and then select a set  $S_p\setminus S$ which is of size at most $15$ since $|S_p\setminus S| = |S\setminus S_p| +|S_p|-|S| \leq 14 + n - |B_p| - (n/2+1) \leq 14 +n - (n/2-2)-(n/2+1)  = 15$.)
\end{proof}

\subsubsection{Putting it All Together}\label{sec:subadditive-proof}

We are now ready to prove the theorem.

\begin{proof}[Proof of Theorem~\ref{thm:subadditive-demand}]
Consider any fixed deterministic algorithm $ALG$, and the random instances described in Section~\ref{sec:subadditive-instances}. 
Let $S_{ALG}$ be the set that the algorithm computes when all demand and value queries are answered according to the function $f$. We claim that on the vast majority of functions $f_T$ the algorithm will also compute $S_{ALG}$.

To this end, let $\mathcal{T}_i$ be the family of all sets $T$ of size $\frac{n}{2}+1$ such that the first $i$ demand or value queries performed by the algorithm would have led to the same answers on $f_T$ as on $f$. Clearly, 
$\vert \mathcal{T}_0 \vert= \binom{n}{n/2+1}$. 
By Lemma~\ref{lem:size-dem} 
any demand query allows us to distinguish at most $n^{29}$ many sets from $f$ (and of course, any value query can distinguish at most one set). That is, $\lvert \mathcal{T}_{i+1} \rvert \geq \lvert \mathcal{T}_i \rvert - n^{29}$.

This implies that $\lvert \mathcal{T}_i \rvert \geq \binom{n}{n/2+1} - i \cdot n^{29}$. 
Furthermore, by Lemma~\ref{lem:onlyT},  the approximation ratio of the algorithm that uses at most $i$ queries is no better than $\frac{20}{\sqrt{n}}$ whenever $T^\star \in \mathcal{T}_i$ and $T^\star \neq S_{ALG}$. This happens with probability of at least
$1- \frac{\lvert \mathcal{T}_i \rvert - 1}{\binom{n}{n/2+1}} $ since $T^\star$ is uniformly distributed in $\mathcal{T}_i$ and the algorithm can choose just one of them.
Overall, we get that the approximation of  an algorithm $ALG$ that uses at most $i$ demand and value queries is at most 
\begin{eqnarray*} 
E\left[\frac{g(S_{ALG})}{g(T^\star)}\right] &\leq & \Pr[T^\star \notin \mathcal{T}_i  \vee T^\star =S_{ALG}] \cdot 1  \\ & + & (1-\Pr[T^\star \notin \mathcal{T}_i  \vee T^\star =S_{ALG}]) \cdot \frac{20}{\sqrt{n}} \\ 
& \leq  &
\frac{\binom{n}{n/2+1}-\lvert \mathcal{T}_i \rvert +1}{\binom{n}{n/2+1}} + \frac{\lvert \mathcal{T}_i \rvert -1}{\binom{n}{n/2+1}} \cdot \frac{20}{\sqrt{n}} \\ 
& \leq & \frac{20}{\sqrt{n}} + \frac{i\cdot n^{29}+1}{\binom{n}{n/2+1}} \leq O\left(\frac{1}{\sqrt{n}}\right),
\end{eqnarray*}
where the last inequality holds for every polynomial amount of demand and value queries.
\end{proof}

We conclude with a simple observation showing that, for subadditive rewards, an $n$-approximation can be achieved using only a value oracle.

\begin{remark}[$n$-approximation for subadditive]\label{rem:subaddtive}
We observe that for subadditive $f$, it is possible to obtain a factor-$n$ approximation with value oracle access by computing the best single-agent contract $g(\{i\})$ for $i \in \agents$.
This follows by observing that, for any set of agents $S \subseteq \agents$ such that $g(S)\ge 0$, it holds that $g(\{i\}) \geq 0$ for all $i \in S$ and $g(S) \leq \sum_{i \in S} g(\{i\}) \leq n \cdot \max_{i \in S}g(\{i\})$. 
To see this, note that for any $i \in S$, we have $f(S) \leq f(S\setminus \{i\}) + f(\{i\})$, so $\frac{c_i}{f(S)-f(S\setminus\{i\})} \geq \frac{c_i}{f(\{i\})}$. Moreover, $f(S) \leq \sum_{i\in S}f(\{i\})$. So, $g(S) = (1-\sum_{i \in S} \frac{c_i}{f(S)-f(S\setminus\{i\}}) f(S) \leq \sum_{i\in S} (1-\sum_{i \in S} \frac{c_i}{f(\{i\})})f(\{i\}) \leq \sum_{i\in S} (1- \frac{c_i}{f(\{i\})})f(\{i\}) = \sum_{i \in S}g(\{i\})$.
\end{remark}

\section*{Acknowledgements}
This project has been partially funded by the European Research Council (ERC) under the European Union's Horizon Europe Program (grant agreement No.~101170373), by an Amazon Research Award, by the NSF-BSF (grant number 2020788), and by the Israel Science Foundation Breakthrough Program (grant No.~2600/24).
Part of this work was done while T. Kesselheim was visiting the Simons Institute for the Theory of Computing.

\bibliographystyle{apalike}
\bibliography{references}

\appendix

\section{Extension of the Model to General Outcome Spaces}\label{app:general-outcomes}

In this appendix, we discuss how to extend our model to general outcome spaces. In this more general setting, the outcome space is $\Omega = \{0,\ldots,m-1\}$ with rewards $r(\omega) \geq 0$ for each $\outcome \in \outcomespace$. Without loss of generality, we may assume that outcomes are indexed so that $r(0) \leq r(1) \leq \ldots \leq r(m-1)$. We assume that $r(0) = 0$ and that $r(m-1) \leq 1$.

As before each agent $i \in \agents$ can either exert effort or not, and the cost of exerting effort is $c_i$ for agent $i$. Each set of agents $S \subseteq \agents$ that exerts effort induces a distribution $q(S)$ over outcomes $\omega \in \Omega$. We use $q_\outcome(S)$ to denote the probability of outcome $\outcome$. 
The distribution $q$ in turn induces a reward function $f: 2^\agents \rightarrow [0,1]$, which  maps each set of agents $S \subseteq \agents$ that exert effort to an expected reward $f(S)$ for the principal. 
Namely, $f(S) = \sum_{\omega \in \Omega} q_\omega(S) \cdot r(\omega)$.
As before, we assume that $f$ is monotonically non-decreasing so that for any $S, S' \subseteq \agents$ with $S \subseteq S'$ it holds that $f(S) \leq f(S')$. We also assume that $f$ is normalized so that $f(\emptyset) = 0$.

For a general contract $t: \Omega \rightarrow \mathbb{R}^n_{\ge 0}$, we define the principal's utility for a set of agents $S$ as $u_P(S,q,t) = f(S) -  \sum_{\outcome \in \outcomespace} q_\outcome(S) \sum_{i \in \agents} t_i(\outcome)$ and agent $i$'s utility as $u_i(S,q,t) = \sum_{\outcome \in \outcomespace} q_\outcome(S) \cdot t_i(\omega) - \indicator{i \in S} c_i$. We define (pure) Nash equilibria as before, i.e., we say that $t$ incentivizes $S$ if
\begin{align*}
&u_i(S,q,t) \geq u_i(S\setminus\{i\},q,t) &&\text{for all $i \in S$, and}\\
&u_i(S,q,t) \geq u_i(S\cup\{i\},q,t) &&\text{for all $i \not\in S$.}
\end{align*}

Generalizing the definition of a linear contract in the binary-outcome case, a linear contract is defined by a vector $\alpha = (\alpha_1, \ldots, \alpha_n) \in \mathbb{R}^n_{\geq 0}$. As before, the interpretation is that $t_i(\omega) = \alpha_i r(\omega)$ for all $i \in \agents$ and $\omega \in \outcomespace$. 

Exploiting that the expected payment of a linear contract $\alpha$ to agent $i$ is $\alpha_i f(S)$, which is independent of the number of outcomes, we obtain the following characterization of sets of agents that can be incentivized and the optimal linear contract for incentivizing a given set of agents. (The proof is analogous to that of Proposition~\ref{prop:binary-outcome}.)

\begin{proposition}
(a) A set of agents $S \subseteq \agents$ can be incentivized by a linear contract if and only if there is no agent $i \in S$ such that $f(i \mid S \setminus \{i\}) = 0$ and $c_i > 0$. (b) Among all linear contracts that incentivize set $S \subseteq A$, the one that maximizes the principal's utility is the following linear contract:
\begin{align*}
    &\alpha_i = \frac{c_i}{f(S) - f(S \setminus \{i\})} = \frac{c_i}{f(i \mid S \setminus \{i\})} &&\text{for all $i \in S$ s.t.~$f(i \mid S \setminus\{i\}) > 0$, and}\\
    &\alpha_i = 0 &&\text{otherwise.}
\end{align*}
\end{proposition}

With this proposition at hand, the principal's problem of designing an optimal linear contract is equivalent to that in the binary-outcome case.

\section{Max-Min Optimality of Linear Contracts With Moment Constraints}\label{app:max-min}

In this appendix, we establish that linear contracts are max-min optimal when only the expected reward of each action is known (extending results of \citet{DuttingRT19,DuttingEFK21}). 

The setting is as follows: We are given the set of agents $\agents$ and their costs $c_i$ for $i \in \agents$. We consider a general outcome space $\Omega = \{0,\ldots,m-1\}$ with fixed rewards. We assume that outcomes are indexed so that $r(0) \leq r(1) \leq \ldots \leq r(m-1) \leq 1$, and that $r(0) = 0$. Importantly, the underlying distribution $q$ over outcomes is \emph{unknown}. We are only given the induced expected reward $f(S)$ for each set of agents $S \subseteq \agents$. As before we assume that $f$ is monotonically non-decreasing and normalized.

For a general contract $t: \Omega \rightarrow \mathbb{R}^n_{\geq 0}$ and a distribution over outcomes $q$, we define the principal's utility for a set of agents $S$ as $u_P(S,q,t) = f(S) - \sum_{\outcome \in \outcomespace} q_\outcome(S) \sum_{i \in \agents}  t_i(\outcome)$ and agent $i$'s utility as $u_i(S,q,t) = \sum_{\outcome \in \outcomespace} q_\outcome(S) t_i(\omega) - \indicator{i \in S} c_i$. We define (pure) Nash equilibria as before, i.e., we say that $t$ incentivizes $S$ if 
\begin{align*}
&u_i(S,q,t) \geq u_i(S\setminus\{i\},q,t) &&\text{for all $i \in S$, and}\\
&u_i(S,q,t) \geq u_i(S\cup\{i\},q,t) &&\text{for all $i \not\in S$.} 
\end{align*}
Let $\mathcal{Q}(\{f(S)\}_{S \subseteq \agents})$ denote all distributions over outcomes that are compatible with the known expected rewards $\{f(S)\}_{S \subseteq \agents}$. For a distribution $q$ over outcomes, let $\textsf{NE}(q,t)$ denote all (pure) Nash equilibria of contract $t$.
We define the principal's worst-case utility for contract $t$ as 
\[
\bar{u}_P(t) = \min_{q \in \mathcal{Q}(\{f(S)\}_{S \subseteq \agents})} \max_{S \in \textsf{NE}(q,t)} u_p(S,q,t).
\]

\begin{proposition}\label{prop:max-min}
In multi-agent contract settings, for any set of known expected rewards $\{f(S)\}_{S \subseteq \agents}$, a linear contract maximizes the principal's worst-case utility.   
\end{proposition}

\begin{proof}
Fix expected rewards $\{f(S)\}_{S \subseteq \agents}$. Let $\alpha^\star$ be the linear contract that maximizes $\bar{u}_P(\alpha)$ over all linear contracts $\alpha$. Note that since linear contracts do not depend on the distribution over outcomes only the expected rewards, the same $S^\star$ will be the optimal Nash equilibrium of $\alpha^\star$ for any compatible distribution, and the principal's utility $u_P(S^\star,q,\alpha^\star)$ for that set of agents is the same across all compatible distributions $q \in \mathcal{Q}(\{f(S)\}_{S \subseteq \agents})$.

Now fix an arbitrary contract $t$. We want to show that $\bar{u}_P(t) \leq \bar{u}_P(\alpha^\star)$. 
To this end, for each $S \subseteq \agents$, let's construct a binary-outcome distribution $\bar{q}(S)$ over $\{0,m-1\}$ that has expected reward $f(S)$ (and is therefore compatible). Namely, let
\begin{align*}
\bar{q}_{\outcome}(S) = 
\begin{cases}
1- \frac{f(S)}{r(m-1))} &\text{if $\outcome = 0$,}\\
0 &\text{if $\outcome \neq 0,m-1$, and}\\
\frac{f(S)}{r(m-1)} &\text{if $\outcome = m-1$.}
\end{cases}
\end{align*}

We then have 
\[
\bar{u}_P(\alpha^\star) = u_P(S^\star,\bar{q},\alpha^\star) \geq \max_{S \in \mathsf{NE}(\bar{q},t)} u_P(S,\bar{q},t) \geq \bar{u}_P(t),
\]
where for the first equality we used that under $\alpha^\star$ the same $S^\star$ is the optimal NE for all compatible distributions and that the principal's utility for that set is the same across all compatible distributions, the first inequality holds by Proposition~\ref{prop:binary-outcome} which shows that for binary-outcome settings linear contracts are optimal, and the final inequality holds by the definition of $\bar{u}_P(t)$ as the minimum of $\max_{S \in \textsf{NE}(q,t)} u_p(S,q,t)$ over all compatible distributions. 
\end{proof}

\section{Additive Reward Functions}
\label{app:additive}


In this appendix, we present our results for additive reward functions. As in the body of the paper, we consider the binary-action and binary-outcome setting, where linear contracts are optimal. However, as observed in Appendix~\ref{app:general-outcomes}, all results in this section generalize to the problem of computing optimal linear contracts in settings with general outcome spaces.

\subsection{NP-Hardness}
\begin{proposition}\label{prop:hardness-for-additive}
The optimal contract problem for additive reward functions over $n$ agents is NP-hard.
\end{proposition}

\begin{proof}
The proof is by reduction from \textsf{PARTITION}. We are given a multi-set  $\{w_1, \ldots, w_n\}$ of positive integers $w_j > 0$ with $\sum_{j = 1}^{n} w_j = W$, and have to decide whether $[n]$ can be partitioned into $I_1$ and $I_2 = [n] \setminus I_1$ such that $\sum_{j \in I_1} w_j = \sum_{j \in I_2} w_j = W/2.$

The corresponding contract problem is: $f$ is additive over $n$ agents, where agent $i$ has a value of $v_i = w_i$ and a cost $c_i = w_i^2/W $. 
Since agent $i$'s marginal is $v_i$, 
the indifference point for agent $i$ is $\alpha_i = c_i/v_i = w_i/W$. The principal's utility for incentivizing a set of 
agents 
$S \subseteq [n]$ is
\[
g(S) = \left(1- \sum_{i \in S} \frac{w_i}{W}\right) \cdot \sum_{i \in S} w_i.
\]
$g$ is maximized when $\sum_{i \in S} w_i$ is closest to $W/2$.
Thus, instance $\{w_1,\ldots,w_n\}$ is a yes-instance, if and only if $g(\sstar) = \frac{W}{4}$.
\end{proof}

\subsection{FPTAS}

\begin{proposition}\label{prop:fptas-for-additive}
There is an FPTAS for the optimal contract problem for an additive reward functions over $n$ agents.
\end{proposition}

Suppose we know $b = \max_{i \in \sstar} f(\{i\})$. We can assume this because there are only $n$ choices and can run the following algorithm with each choice.

Let $\delta = \frac{\epsilon}{n}$ and define an additive function $\tilde{f}$ by rounding down each $f(\{i\})$ to the next multiple of $\delta b$. That is, $\tilde{f}(\{i\}) = \left\lfloor \frac{f(\{i\})}{\delta b} \right\rfloor \delta b$ and $\tilde{f}(S) = \sum_{i \in S} \tilde{f}(\{i\})$. This way, each $\tilde{f}(S)$ is a
multiple of $\delta b$.

Let $T_x$ be the set $S$ that minimizes $\sum_{i \in S} \frac{c_i}{f(\{i\})}$ subject to $\tilde{f}(T) \geq x$. Our algorithm returns the set $T_x$ that maximizes $\left( 1 - \sum_{i \in T_x} \frac{c_i}{f(\{i\})}\right) x$ among all $x = k \delta b$ for $k \in \{0, 1, \ldots, \lceil \frac{n}{\delta} \rceil\}$.

\begin{claim}
The algorithm can be implemented in polynomial time in $n$ and $\frac{1}{\epsilon}$.
\end{claim}

\begin{proof}
Observe that we only need to compute polynomially many sets $T_x$. Furthermore, these sets can be computed by dynamic programming in polynomial time.

To this end, let $$A(j, x) = \min\left\{ \sum_{i \in S} \frac{c_i}{f(\{i\})} \mid S \subseteq \{1, \ldots, j\}, \tilde{f}(S) \geq x \right\},$$  $A(0, x) = 0$ if $x \leq 0$ and $A(0, x) = \infty$ if $x > 0$. Observe that
\[
A(j, x) = \min\left\{A(j-1,x), A(j-1, x-\tilde{f}(\{i\})) + \frac{c_i}{f(\{i\})}\right\},
\]
which completes the proof.
\end{proof}

\begin{claim}
The algorithm computes a $(1-\epsilon)$-approximation {to the optimal contract}.
\end{claim}

\begin{proof}
Note that for every set $S$, we have $\tilde{f}(S) \leq f(S)$. Therefore for every set $T_x$
\[
g(T_x) = \left( 1 - \sum_{i \in T_x} \frac{c_i}{f(\{i\})}\right) f(T_x) \geq \left( 1 - \sum_{i \in T_x} \frac{c_i}{f(\{i\})}\right) x.
\]

In particular, for the set $T_x$ that we return, we have
\begin{eqnarray*}
\left( 1 - \sum_{i \in T_x} \frac{c_i}{f(\{i\})}\right) x  & \geq & \left( 1 - \sum_{i \in T_{\tilde{f}(\sstar)}} \frac{c_i}{f(\{i\})}\right) \tilde{f}(\sstar) \\ &\geq & \left( 1 - \sum_{i \in \sstar} \frac{c_i}{f(\{i\})}\right) \tilde{f}(\sstar)
\end{eqnarray*}
because $\tilde{f}(\sstar) = k \delta b$ for some $k$ in the range.

Finally, observe that we have
\begin{eqnarray*}
\tilde{f}(\sstar) & = & \sum_{i \in \sstar} \tilde{f}(\{i\}) \geq \sum_{i \in \sstar} (f(\{i\}) - \delta b) \\ & {\geq} & f(\sstar) - n \delta b \geq (1-\epsilon) f(\sstar).
\end{eqnarray*}
Therefore
\[
(1-\epsilon) g(\sstar) \leq \left( 1 - \sum_{i \in \sstar} \frac{c_i}{f(\{i\})} \right) \tilde{f}(\sstar),
\]
{as claimed.}
\end{proof}

\section{Submodular with Value Queries}
\label{app:extension-to-submodular}

In this appendix, we show how to adjust the proof of Theorem~\ref{thm:xos} in Section~\ref{sec:sm-and-xos} to submodular reward functions with value queries.

\subsection{Approximate Demand Query via Value Queries}

Our adjusted argument relies on the ability to compute an approximate demand set as formalized in the following lemma.

\begin{lemma}[\citet{SviridenkoVW17,Harshaw19a}]
\label{lem:approx-demand}
For submodular $f$ given access to a value oracle, there exists a poly-time algorithm that finds a set $S$  such that
\[
f(S) - \sum_{i \in S} p_i \geq (1-1/e)f(T) - \sum_{i \in T} p_i \quad \text{for all $T$.}
\]
\end{lemma}

We call the set $S$ a $\beta$-approximate demand set, where $\beta = 1-1/e$. Note that this notion of an approximate demand query is not a fully multiplicative approximation, it's weaker in that we need to subtract the full price of a set.

\subsection{Adjusting the Proof to Approximate Demand Queries}

We next show how to adapt our proof to the case of submodular $f$, using only value oracle calls.
To do so, we first show that if we modify Algorithm~\ref{alg:xos-contract-estimate} so that (1) we  initialize  $\downscaling$ to be $\downscaling=\frac{\beta^2 \cdot \optestimate}{32}-\max_{i\in \actions'} f(i)$, (2) we initialize  the prices to be $p_i=\frac{\beta}{2}\cdot \sqrt{c_i\cdot\optestimate}$, (3) we replace the initialization of $T$ to be a $\beta$-approximate demand set instead of an exact demand set (which can be done for $\beta = 1-\frac{1}{e}$ using value oracle by Lemma~\ref{lem:approx-demand}), 
and (4) we remove actions with negative utility from $T$, then the algorithm returns  a set $U$ with the following guarantee:

\begin{lemma}
Algorithm~\ref{alg:submod-contract-estimate} runs in polynomial time with access to a demand and value oracle to $f$. 
If $\optestimate \leq f(\sstar \cap A' )$, then the set $U \subseteq \actions'$ that it returns satisfies 
\[
g(U) \geq \max\left\{\frac{\optestimate}{128}  - \frac{\max_{i \in A'} f(\{i\})}{4},0\right\}. 
\]
\label{lem:gu128sub}
\end{lemma}

\begin{algorithm}[t]
\caption{Approximately optimal contract given an estimate of $f(\sstar)$}\label{alg:submod-contract-estimate}
\KwData{An XOS function $f:2^\actions \rightarrow \reals_{\geq 0}$, costs $\{c_{i}\}_{i\in \actions}$, an estimate $\optestimate$, a parameter $\beta$}

\KwResult{A set $U \subseteq \actions'$}

    Let $\downscaling=\frac{\beta^2 \cdot \optestimate}{32} - \max_{i \in \actions'} f(i)$.

    For every $i\in \actions'$, let $p_{i} = \frac{\beta}{2} \cdot \sqrt{ c_i \cdot \optestimate }$ ($p_{i} =\infty$ for $i\in \actions\setminus\actions'$)

    Let $T $ be a $\beta$-demand set with prices $p_{i}$ over the set $\actions'$
    
    \While{ Exists $i\in T$ such that $f(i \mid T\setminus\{i\} ) < p_i$}{ $ T = T \setminus\{i\}$}

    \eIf{$0<\downscaling<f(T)$}{

    $U \leftarrow $  the output of Algorithm~\ref{alg:scaling} on $(f, T, \downscaling, \delta= \frac{1}{2})$
    }
    {
    $U\leftarrow \emptyset$ 
    }
    
\Return{U}
        
\end{algorithm}

\begin{proof}
If $\downscaling\leq 0$ then the algorithm returns the empty set, which satisfies the proof. 
Otherwise, if we denote by $T^1$ the $\beta$-demand set, and by $T^2$ the $\beta$-demand set after removing actions with negative utility, then
we have
\begin{eqnarray}
f(T^2) & \geq & f(T^2) - \sum_{i \in T^2} p_i  \nonumber \\
& \geq & f(T^1) - \sum_{i \in T^1} p_i  \nonumber \\
& \geq & {\beta}f(\sstar \cap A') - \sum_{i \in \sstar \cap A'} p_i \nonumber \\
& = & {\beta}f(\sstar \cap A') - \frac{{\beta}\sqrt{\optestimate}}{2} \cdot \sum_{i \in \sstar \cap A'} \sqrt{c_i} \nonumber \\
&\geq &  {\beta}f(\sstar \cap \actions' ) -{\frac{\beta}{2}} \sqrt{f(\sstar \cap \actions' ) \cdot  \optestimate}\nonumber \\ 
& \geq &  {\frac{\beta}{2}} \cdot f(\sstar \cap \actions') , \label{eq:ft-sub}
\end{eqnarray}
where the first inequality is since the prices are non-negative, the second inequality is since we remove actions with negative utilities, the third inequality since $T^1$ is {an approximate} demand set over the set $A'$, the fourth inequality is by Lemma~\ref{lem:stronger-sumci}, and the last inequality holds because $\optestimate \leq f(\sstar \cap \actions')$. 
Furthermore, as in $T^2$ we deleted all actions with negative utility, we have to have $f(i \mid T^2 \setminus\{i\}) \geq p_i$ for all $i \in T^2$.

Thus, since $\optestimate \leq f(\sstar \cap \actions')$ by assumption and {$\downscaling = \beta^2 \optestimate/32 - \max(0,\max_{i \in \actions'} f(\{i\})) < \beta^2 \optestimate/2$} it holds that 
$$f(T^2)\stackrel{\eqref{eq:ft-sub}}{\geq} \frac{{\beta}}{2}\cdot f(\sstar \cap \actions') \geq  \frac{{\beta}\optestimate}{2} 
> \downscaling.$$
By Lemma~\ref{lem:xosscaling}, 
the set $U \subseteq T^2$ fulfills 
$$f(U) \geq \left(1-\delta\right) \downscaling = \frac{{\beta^2}}{64} \optestimate - \frac{1}{2} \max_{i \in T^2} f(\{i\}) \geq \frac{{\beta^2}}{64} \optestimate - \frac{1}{2} \max_{i \in \actions'} f(\{i\}).$$
Furthermore, $U$ fulfills 
\begin{equation}
    f(U) \leq \downscaling + \max_{i \in T^2} f(\{i\}) \leq  \frac{{\beta^2}\optestimate}{32}\label{eq:fu-up-sub}
\end{equation} 
and  for all $i \in U$ it holds 
\begin{equation}
f(i \mid U \setminus\{i\}) \geq \delta f(i \mid T^2 \setminus\{i\}) = \frac{1}{2} \cdot  f(i \mid T^2 \setminus\{i\}) \label{eq:fu-marg-sub}.    
\end{equation}
 Therefore, for all $i \in U$, we have
\begin{eqnarray*}
    f(i \mid U \setminus\{i\}) 
    & \stackrel{\eqref{eq:fu-marg-sub}}{\geq} & \frac{1}{2} \cdot f(i \mid T^2 \setminus\{i\}) 
    \geq \frac{1}{2}\cdot  p_i \\ & = & \frac{{\beta}}{4}\cdot \sqrt{c_i \optestimate} = \sqrt{2 {\beta^2}c_i \frac{\optestimate}{32}} \stackrel{\eqref{eq:fu-up-sub}}{\geq} \sqrt{2 c_i f(U)}.
\end{eqnarray*}
So, Lemma~\ref{lem:stronger-condition} implies $g(U) \geq \frac{1}{2} f(U)$. So, in combination, 
\[
g(U) \geq \frac{1}{2}f(U) \geq \frac{1}{2} \left( \frac{{\beta^2}}{64} \optestimate - \frac{1}{2} \max_{i \in \actions'} f(\{i\}) \right),
\]
as claimed.
\end{proof}

We next show that Algorithm~\ref{alg:xos-contract} that uses as a subroutine the modified algorithm (Algorithm~\ref{alg:submod-contract-estimate})  fulfills Theorem~\ref{thm:xos}.
\begin{proof}[Proof of Theorem~\ref{thm:xos} (submodular rewards)]
If $f(\sstar\cap \actions') \leq {\frac{128 \xi}{\beta^2}} \cdot \max\{0,\max_{i\in \actions} g(\{i\}\}$, then by Lemma~\ref{lem:onebigagent} the best single action (or the empty set) gives at least an approximation of $\frac{\beta^2}{{128\xi+ \beta^2}}$. Else, $\actions' \neq \emptyset$ (thus $\downscaling$ is well defined), and  it holds that 
\[f(\sstar\cap\actions')  \geq {\frac{128 \xi}{\beta^2}} \cdot \max\{0,\max_{i\in \actions'} g(\{i\})\} \geq \max_{i\in \actions'} f(\{i\})/2\geq \downscaling,
\]
and  $ f(\sstar\cap \actions') \leq 2n \cdot \downscaling$.

Let $j^\star$ be the unique $j$ in which $\optestimateround{j} \leq f(\sstar \cap \actions') < \xi \cdot \optestimateround{j}$.  Note that because of the bounds we just established we must have $j^\star\geq 0$ and $j^\star\leq \lceil\log_{\xi} 2n\rceil$. Now, since $ \frac{1}{\xi} \cdot f(\sstar \cap \actions') \leq \optestimateround{j} \leq f(\sstar \cap \actions')$,
\begin{eqnarray*}
g(U^{(j^\star)})    & \geq &  \frac{{\beta^2}\optestimateround{j^\star}}{128} - \frac{\max_{i\in \actions'} f(\{i\})}{4}  \\   
&\geq & \frac{{\beta^2} f(\sstar\cap \actions')}{128\cdot \xi} - \frac{\max_{i\in \actions} g(\{i\})}{2}  \\& \geq &  {\beta^2} f(\sstar\cap \actions')\cdot \left(\frac{1}{128\cdot \xi} - \frac{1}{256\cdot \xi} \right) \\  & \geq &
{\beta^2}\frac{{\frac{128\xi}{\beta^2}} \cdot g(\sstar)}{{\frac{128\xi}{\beta^2}}+1}\frac{1}{256\cdot \xi} = \frac{\beta^2 \cdot g(\sstar)}{{256\xi+ 2\cdot \beta^2}},
\end{eqnarray*}
where the first inequality is by Lemma~\ref{lem:gu128sub}, the second inequality is since for every $i\in \actions'$ it holds that $\frac{c_i}{f(\{i\})}\leq \frac{1}{2}$ and thus $g(\{i\}) = f(\{i\})(1-\frac{c_i}{f(\{i\})}) \geq f(\{i\})/2$, the third and fourth inequalities are since we are considering the case that $f(\sstar\cap \actions') > {\frac{128 \xi}{\beta^2}} \cdot \max\{0,\max_{i\in \actions} g(\{i\})\}$.
This  gives an approximation of {$\frac{\beta^2}{256\xi+2\cdot \beta^2}$}.
\end{proof}

\section{Inapproximability Result for Submodular Reward Functions}
\label{app:upper-submod}

In this appendix, we present additional details for our impossibility result for submodular reward functions in Section~\ref{sec:submodular-upper}. In Section~\ref{sec:fTsubmod}, we show that the function $f_T$ is monotone submodular. In Section~\ref{sec:fTnotGS}, we define the class of gross-substitutes functions, a special case of submodular functions, and show that $f_T$ does not belong to that class.

\subsection{Function \texorpdfstring{$f_T$}{f\_T} is Monotone and Submodular}\label{sec:fTsubmod}

\begin{claim}
    For every set $T\subset A$ of even size, the function $f_T$ is monotone and submodular.
    \label{cl:submodularity}
\end{claim}

\begin{proof} 
Consider the discrete partial derivatives of $f$ with respect to $x$ and $y$; i.e., $f_x(x,y; k)= f(x+1,y; k)-f(x,y; k)$ and $f_y(x,y; k)= f(x,y+1; k)-f(x,y; k)$, given by
\[
f_x(x,y; k)
= \begin{cases}
    0 & \text{ if $x+y \geq k$} \\
        \frac{1}{\sqrt{x+y+1} + \sqrt{x+y}} & \text{ if $x+y < k$, $y \leq \frac{k}{2}$} \\
    \frac{1}{\sqrt{x +1+ \frac{k}{2}}+\sqrt{x + \frac{k}{2}}}  +  \frac{y- \frac{k}{2}}{\sqrt{k} + \sqrt{x +1+ \frac{k}{2}}} -  \frac{y- \frac{k}{2}}{\sqrt{k} + \sqrt{x + \frac{k}{2}}}  & \text{ otherwise}
\end{cases}
\]
and 
\[
f_y(x,y; k) 
= \begin{cases}
    0 & \text{ if $x+y \geq k$} \\
    \frac{1}{\sqrt{x+y+1} + \sqrt{x+y}} & \text{ if $x+y < k$, $y < \frac{k}{2}$} \\
     \frac{1}{\sqrt{k} + \sqrt{x + \frac{k}{2}}}  & \text{ otherwise.}
\end{cases}
\]
To show that $f_T$ is monotone, it is equivalent to prove that $f_x(x,y;k),f_y(x,y;k)\geq 0$ for every $x,y$.
The only non-trivial case is the ``otherwise'' case of $f_x$, which includes a negative term. For this case, we have
\begin{eqnarray*}
f_x(x,y; k) & = &   \frac{1}{\sqrt{x +1+ \frac{k}{2}}+\sqrt{x + \frac{k}{2}}}  + \underbrace{  \frac{y- \frac{k}{2}}{\sqrt{k} + \sqrt{x +1+ \frac{k}{2}}} -  \frac{y- \frac{k}{2}}{\sqrt{k} + \sqrt{x + \frac{k}{2}}}}_{\leq  0} \\ & \geq &  
\frac{1}{\sqrt{x +1+ \frac{k}{2}}+\sqrt{x + \frac{k}{2}}}  +  \frac{ \frac{k}{2}-x}{\sqrt{k} + \sqrt{x +1+ \frac{k}{2}}} -  \frac{\frac{k}{2}-x}{\sqrt{k} + \sqrt{x + \frac{k}{2}}} \\ & = & \frac{1}{\sqrt{k} + \sqrt{1 + \frac{k}{2} + x}} \geq 0,
\end{eqnarray*}
where the first inequality holds since in this case, $x+y \leq k$.

To show submodularity, it is sufficient to show that the continuous derivatives of all cases of $f_x$ and $f_y$ with respect to both $x,y$ are non-positive, and that $f_x$ and $f_y$ 
are non-increasing
in the transitions between the cases. 
The only non-trivial case for the non-positivity of the respective derivatives is $f_x(x,y; k)/\delta x$ in the ``otherwise'' case of $f_x$ for which we have

\begin{eqnarray*}
\frac{f_x(x,y; k)}{\delta x} & = & 
\underbrace{\frac{-1}{\sqrt{k/2 + x} \sqrt{2 + k + 2 x} (\sqrt{k + 2 x} + \sqrt{2 + k + 2 x})}}_{\textit{increasing in } x} 
\\ & + & \underbrace{\frac{y-\frac{k}{2} }{2 \sqrt{\frac{k}{2} + x} (\sqrt{k} + \sqrt{\frac{k}{2} + x})^2} - \frac{y-\frac{k}{2} }{2 \sqrt{1 + \frac{k}{2} + x} (\sqrt{k} + \sqrt{1 + \frac{k}{2} + x})^2}}_{\textit{decreasing in } x \textit{ and increasing in } y}\\ & \leq & 
-\frac{1}{\sqrt{k} \sqrt{2k} (\sqrt{2k } + \sqrt{2k})} +\frac{\frac{k}{2} }{2 \sqrt{\frac{k}{2} } (\sqrt{k} + \sqrt{\frac{k}{2} })^2} - \frac{\frac{k}{2} }{2 \sqrt{1 + \frac{k}{2} } (\sqrt{k} + \sqrt{1 + \frac{k}{2} })^2} \leq 0,
\end{eqnarray*}
where for the first inequality we use that $x\in[0,\frac{k}{2}]$ and that $y\leq k$, and the last inequality holds for every $k>0$.

It remains to show that $f_x$ and $f_y$ are non-increasing in the transitions between cases. 
The possible transitions are from the second to first, third to first, and second to third cases. 
Transitions to the first case are clearly non-increasing, as they go from non-negative values to zero. We next show that any transition from the second case to the third one is non-increasing.
That is, we need to show that  (1) $f_x(x,k/2; k) \geq f_x(x,k/2+1; k) $ for all $x \leq k/2$, and that (2) $f_y(x,k/2-1; k) \geq f_y(x,k/2; k)$ for all $x\leq k/2$.
The first inequality holds since 
\begin{eqnarray*}
f_x(x,k/2; k)  & = & \frac{1}{\sqrt{x +1+ \frac{k}{2}}+\sqrt{x + \frac{k}{2}}} \\ &  \geq &  \frac{1}{\sqrt{x +1+ \frac{k}{2}}+\sqrt{x + \frac{k}{2}}}  +  \underbrace{\frac{1}{\sqrt{k} + \sqrt{x +1+ \frac{k}{2}}} -  \frac{1}{\sqrt{k} + \sqrt{x + \frac{k}{2}}}}_{ \leq 0 } =  \;f_x(x,k/2+1; k), 
\end{eqnarray*}
and the second inequality holds since 
\begin{eqnarray*}
f_y(x,k/2-1; k) & = & \frac{1}{\sqrt{x+\frac{k}{2}} + \sqrt{x+\frac{k}{2}-1}} \geq  \frac{1}{\sqrt{k} + \sqrt{x + \frac{k}{2}}} =  \;f_y(x,k/2; k), 
\end{eqnarray*}
where the inequality is since $x \leq \frac{k}{2}$.
\end{proof}

\subsection{Function \texorpdfstring{$f_T$}{f\_T} is not Gross Substitutes}\label{sec:fTnotGS}

\begin{definition}
A set function $f$ is \emph{gross substitutes} if for any two vectors $p, q \in \reals_+^n$ and any $S \subseteq \actions$ such that 
$S \in \arg\max_{S' \subseteq \actions} (f(S') - \sum_{j \in S'} p_j)$ there exists a set $T \subseteq \actions$ such that 
$T \in \arg\max_{T' \subseteq \actions} (f(T') - \sum_{j \in T'} q_j)$ and $\{j \in S \mid q_j \leq p_j\} \subseteq T$.   
\end{definition}

\begin{observation}
For every $T\subset \actions$ of even size $k\geq 4$, the function $f_T$ is not gross substitutes.  
\label{obs:not-gs}
\end{observation}

\begin{proof}
According to the triplet condition in \citep{ReijniersePG02}, a function $h:2^\actions\rightarrow \reals_{\geq 0}$ is gross substitutes if and only if it is submodular and for any set $S$, and any three elements $\{i,j,k\} \not\in S$, it holds that
\[
\max\{h(i,k \mid S) + h(j \mid S), h(j,k \mid S) + h(i \mid S)\} \geq h(i,j \mid S) + h(k \mid S).
\]
Consider two arbitrary elements $i_1,i_2\in T$, and arbitrary element $j\not\in T$, and let $S \subseteq T \setminus \{i_1,i_2\}$ be an arbitrary set of size $k/2$.

The triplet condition requires that
\[
\max\{f_T(i_1,j \mid S) + f_T(i_2 \mid S), f_T(i_2,j \mid S) + f_T(i_1 \mid S)\} \geq f_T(i_1,i_2 \mid S) + f_T(j \mid S),
\]
which is equivalent to
\[
   f(1, \frac{k}{2}+1; k) + f(0, \frac{k}{2} + 1; k) \geq f(1, \frac{k}{2}; k) + f(0, \frac{k}{2} + 2; k).
\]
However,
\begin{eqnarray*}
f(1, \frac{k}{2}+1; k) + f(0, \frac{k}{2} + 1; k) & = & 
\sqrt{\frac{k}{2}+1} + \frac{1}{\sqrt{k} + \sqrt{\frac{k}{2}+1}}+  \sqrt{\frac{k}{2}} + \frac{1}{\sqrt{k} + \sqrt{\frac{k}{2}}}\\ 
& < & 
\sqrt{\frac{k}{2}+1} + \sqrt{\frac{k}{2}} + \frac{2}{\sqrt{k} + \sqrt{\frac{k}{2}}} \\
& = & 
 f(1, \frac{k}{2}; k) + f(0, \frac{k}{2} + 2; k), \end{eqnarray*}
 which concludes the proof.
\end{proof}

\end{document}